\documentclass[letterpaper,journal]{IEEEtran}
\usepackage{amsmath,amsfonts}
\usepackage{amssymb}
\usepackage{algorithmic}
\usepackage{algorithm}
\usepackage{array}
\usepackage{textcomp}
\usepackage{stfloats}
\usepackage{url}
\usepackage{verbatim}
\usepackage{graphicx}
\usepackage{cite}
\usepackage{tabularx}
\usepackage{siunitx}
\usepackage{caption}
\usepackage{hyperref}
\usepackage{xcolor}
\usepackage{subcaption}
\newtheorem{thm}{Theorem}

\begin{document}

\title{A Covert Precision Satellite Communication Framework
Assisted by Cooperative IRSs}

\author{Haoyang~Wu,~\IEEEmembership{Graduate Student Member,~IEEE}, Yunfan~Bai,~\IEEEmembership{Graduate Student Member,~IEEE}, Mei~Shen, Yuwen~Qian, Guangji~Chen, Long Shi,~\IEEEmembership{Senior Member,~IEEE}, Feng~Shu,~\IEEEmembership{Senior Member,~IEEE}, and Jun~Li,~\IEEEmembership{Fellow,~IEEE}

\thanks{Haoyang~Wu, Yunfan~Bai, Mei~Shen, Yuwen~Qian, Guangji~Chen and Long~Shi are with the School of Electronic and Optical Engineering,
Nanjing University of Science and Technology, Nanjing 210094, China (e-mail:
\{endlessing,baiyunfan,shenmei,admon\}@njust.edu.cn,
\protect\nolinkurl{guangjichen@um.edu.mo},
\protect\nolinkurl{slong1007@gmail.com}).}

\thanks{Feng~Shu is with the School of Information and Communication Engineering and Collaborative Innovation Center of Information Technology, Hainan University, Haikou 570228, China, and also with the School of Electronic and Optical Engineering, Nanjing University of Science and Technology, Nanjing 210094, China. (e-mail: shufeng0101@163.com).}

\thanks{Jun~Li is with the School of Information Science and Engineering, Southeast University, Nanjing 210096, China. (e-mail: jun.li@seu.edu.cn).}

}

\maketitle

\begin{abstract}
Satellite communication (SatCom), as an effective complement to terrestrial networks, has attracted considerable attention from both academia and industry owing to its wide coverage and high flexibility.
However, the inherent openness of satellite links renders them highly vulnerable to eavesdropping, thereby posing significant security challenges.
In this paper, we propose a satellite covert precision wireless communication (CPWC) system, where multiple intelligent reflecting surfaces (IRSs) cooperate to assist satellite transmissions, ensuring that confidential information is delivered to legitimate users while remaining undetectable to wardens. 
To further enhance covertness, an orthogonal frequency division multiplexing (OFDM)-based random subcarrier selection (RSCS) method is developed to concentrate the signal energy at the intended receiver.
Under a practical satellite--terrestrial channel model, we derive closed-form covertness constraints for the CPWC system based on relative entropy and detection error probability.
Under the relative-entropy constraint and the satellite power constraint, we maximize the covert rate by an alternating-optimization (AO) based semidefinite relaxation (SDR)  iterative algorithm and obtain a high-quality feasible solution.
Using this solution as a warm start, we further impose the detection-error-probability constraint and refine the beamformer through a sequential quadratic programming (SQP) based algorithm.
Numerical results demonstrate the effectiveness of the proposed CPWC system, where the detection-error-probability-based scheme outperforms the second-order cone programming (SOCP) benchmark, the random-phase-shift design, the one-bit IRS quantized scheme, and the SDR baseline without precise communication (PC) in terms of covert rate.
\end{abstract}


\section{Introduction}
\begingroup
\rightskip=0pt plus 1em
\IEEEPARstart{W}{ith} the increasing communication demands and con\-tinuous advances in communication technologies, satellite communication (SatCom) has emerged as an indispensable complement to terrestrial networks, owing to the advantages of wide coverage and high flexibility~\cite{zhu2021integrated}.
Particularly, low Earth Orbit (LEO) satellites, operating at altitudes below 1500~km, have gained significant attention in recent years due to their benefits of reduced propagation delay, low deployment cost, and mitigated signal attenuation~\cite{8836603}.
\par
\endgroup

LEO satellite systems are now widely deployed for environmental monitoring, maritime communications, emergency response, and broadband access in remote regions~\cite{9832117}, yet the broadcast and open nature of wireless channels makes their links highly vulnerable to eavesdropping and raises stringent security requirements~\cite{8850067}.

Typical approaches to ensure SatCom security include message encryption~\cite{1561945} and physical-layer security (PLS)~\cite{8850067,zhang2023survey}.
Generally,  encryption methods rely heavily on secret keys to protect confidential messages, which inevitably brings additional communication overhead.
Consequently, many SatCom systems are expected to achieve a critical tradeoff between security and latency~\cite{qi2019secure,Wang2024JCINQoS}.
PLS exploits the intrinsic characteristics of the wireless channel to safeguard communications and thus is regarded as a promising complement to the upper-layer cryptographic methods~\cite{6739367}.
To enhance PLS in satellite networks, various strategies have been proposed.
For instance, a beamforming scheme combined with terrestrial interference was applied in a cognitive satellite--terrestrial network~\cite{8353865}.
Different from~\cite{8353865}, a non-cooperative beamforming scheme was investigated, where the secrecy rate is maximized via adaptive beamforming, artificial noise, and maximum ratio transmission at terrestrial base stations~\cite{8333695}.
More recently, UAV-enabled channel reconfiguration has been investigated by combining rotatable antennas with UAV platforms, where antenna rotation and UAV deployment are jointly exploited to reshape wireless propagation for integrated sensing and communication~\cite{chen2026rotatable}.

Nevertheless, the aforementioned security schemes can merely guarantee that the encrypted information cannot be accessed, but these schemes cannot ensure the undetectability of the wireless transmissions.
To address this limitation, covert communication has recently attracted significant attention as a means to enhance communication security~\cite{Li2025IntelligentCovertCommunication}.
Initially, a multi-carrier covert SatCom architecture based on electromagnetic environment sensing was proposed in~\cite{8946585} for a military robot swarm, which uses polar-code-aided multi-stream synchronization and diversity combining to approach ideal synchronization.
Beyond military applications, covert SatCom is also important in civilian scenarios.
Covert communication schemes in satellite downlinks were investigated in~\cite{xu2022covert,wu2022covert}, where massive multiple-input multiple-output (MIMO) hybrid beamforming and relay-assisted transmission with partial relay selection were exploited for millimeter-wave covert transmission.

Although the aforementioned strategies can effectively enhance covert performance, the applicability is often limited by the randomness of wireless propagation, thereby limiting robustness under diverse channel conditions.
In covert SatCom, these challenges are further exacerbated by severe fading, shadowed regions without direct communication links, and long transmission distances, which yield low signal-to-noise ratios (SNRs) at receivers.
To address these issues, intelligent reflecting surfaces (IRSs) have emerged as a promising technology for SatCom.
An IRS is composed of numerous low-cost passive reflecting elements that can dynamically manipulate incident signals in terms of phase, amplitude, frequency, and polarization~\cite{11007277}.
By intelligently reconfiguring the propagation environment, IRSs can significantly improve communication quality~\cite{Chen2023ActiveIRS}.
Furthermore, IRSs can be deployed on walls, building facades, ceilings, and other infrastructural surfaces, a property that facilitates large-scale deployment.
For instance, a two-sided IRS-assisted LEO SatCom system was investigated, where IRSs are deployed at both the satellite and ground nodes to facilitate the joint design of active and passive beamforming for enhanced channel gain~\cite{9849035}.
In~\cite{10483088}, deep reinforcement learning was used to optimize a hybrid IRS in a satellite downlink and maximize the worst-case secrecy rate. Moreover, an IRS-enhanced satellite covert communication system was investigated in~\cite{10035943}, where joint optimization of transmit precoding and IRS reflections significantly improved covert performance.

Although the deployment of IRS enhances communication quality for legitimate users, the same configuration can inadvertently increase the detection probability of a warden located in the same angular direction and thus cause covert communication failures. 

In existing studies, precise communication (PC) has been implemented by using frequency diverse arrays, directional modulation, and beamforming designs~\cite{Shu2021EnhancedSecrecyIRS,Zou2025DirectionalModulationRIS}.
By extending conventional beamforming from angle selectivity to joint angle-range selectivity, PC enables signal energy to be concentrated on a prescribed spatial point rather than an entire angular beam.
Therefore, PC is employed to overcome the limitation of angle-domain beamforming when different spatial points have the same angular direction.
In the proposed CPWC framework, PC reshapes the spatial energy distribution before the warden performs detection: the intended signal is focused around the legitimate user, while the signal leakage toward the warden is reduced.

For example, a random frequency diverse array-based directional modulation (DM) with an artificial noise scheme was developed to achieve secure transmission in both angle and range domains~\cite{7817778}. Furthermore, a secure and precise wireless transmission framework was proposed in~\cite{8333706}, which jointly employs artificial-noise projection, phase-aligned beamforming, orthogonal frequency division multiplexing (OFDM)-based random subcarrier selection (RSCS), and DM~\cite{11030570}.

From the perspective of spatial energy distribution, PC and covert communication are coupled through the signal leakage observed at the warden. 
In covert communication systems, the warden-side received signal leakage under the transmission hypothesis is required to remain close to that under the non-transmission hypothesis, whereas PC reduces signal leakage by concentrating the signal energy around the legitimate user.
When the warden lies in the same angular direction as the legitimate user,
covert communication assisted by single-IRS may enhance the signal power at both the user and the warden.
In the proposed CPWC framework, the cooperative IRSs provide additional phase-control degrees of freedom, allowing the reflected signal components to be constructively combined at the legitimate user while being non-coherently combined at the warden.
Therefore, with the assistance of multiple IRSs, PC provides location-selective spatial control before the warden performs statistical detection~\cite{shen2019two}.
In parallel, covert communication quantifies whether the residual warden-side observation remains close to the non-transmission hypothesis.

To the best of our knowledge, the concept of PC has not yet been explored in the context of covert communications, which leaves the design of PC mechanisms for covert satellite networks largely open. Motivated by this gap, this paper has investigated an OFDM-based RSCS covert communication architecture assisted by distributed IRSs in a satellite covert precision wireless communication (CPWC) system. 
Although a single large IRS can provide a substantial array gain under symmetric channel conditions~\cite{9427474}, distributed IRSs offer improved coverage, robustness to blockages, and finer spatial control in complex terrestrial environments with geographically dispersed users.

In particular, the contributions of this paper are outlined as follows.

\begin{itemize}
\item We propose a novel distributed multi-IRS-aided satellite CPWC system, 
thereby enabling location-selective spatial focusing around the legitimate user while suppressing signal leakage toward the warden, especially when the warden lies in the same angular direction as the legitimate user.

\item A practical satellite-terrestrial channel model is established, where the satellite–IRS and IRS–user links are modeled by Nakagami-m and Rician fading channels, respectively.
With the established channel model,the closed-form expression of the covertness threshold has been derived, which is used to evaluate the detection performance of the warden.

\item For a three-dimensional deployment scenario, closed-form expressions of the channel phase are derived for the satellite-IRS and IRS-user links by employing the OFDM-based random subcarrier selection.

\item 
To maximize the covert rate, an alternating-optimization (AO) based semidefinite relaxation (SDR) algorithm is designed under the relative-entropy constraint to obtain a high-quality warm start. Subsequently, a sequential quadratic programming (SQP) algorithm is designed, which is initialized by the warm start to further improve the covert rate.
\end{itemize}

The remainder of this paper is organized as follows. Section II introduces the distributed multi-IRS-aided satellite CPWC system model. Section III formulates a relative-entropy-con\-strained design to obtain a high-quality feasible solution and develops an AO-based SDR iterative algorithm that serves as a warm start. Building on this initialization, Section IV adopts the detection error probability as the covertness metric and proposes an SQP-based refinement algorithm to produce the final beamforming design. Section V presents numerical results to verify the effectiveness of the proposed scheme. Finally, Section VI concludes the paper.

\textit{Notations:} Bold lowercase and uppercase letters denote vectors and matrices, respectively.
$\mathbb{C}^{x\times y}$ represents an $x\times y$ complex-valued matrix.
$j$ denotes the imaginary unit.
For a complex-valued vector $\mathbf{x}$, $\left \| \mathbf{x}\right \|$ denotes its Euclidean norm, and $\mathrm{diag}(\mathbf{x})$ is the diagonal matrix with each element of $\mathbf{x}$ as a diagonal entry.
For a square matrix $\mathbf{X}$,~$\text{tr}(X)$ denotes its trace, and $\mathbf{X}\succeq \mathbf{0}$ means that the matrix is positive semidefinite.
$~\text{rank}[\cdot ]$, $[\cdot ]^T$, and $[\cdot ]^{H}$ represent the rank, transpose, and conjugate transpose operations, respectively.
$\Re[\cdot ]$ and $\Im[\cdot ]$ denote the real and imaginary parts, respectively.
$\odot$ denotes the piecewise product of two vectors of the same dimension. $\mathcal{O}(\cdot)$ denotes the asymptotic upper bound on the order of computational complexity. $\mathbf{x}\sim\mathbb{C}\mathbb{N}(\boldsymbol{\mu}, \boldsymbol{\Sigma})$ denotes that the complex random vector $\mathbf{x}$ follows a circularly symmetric complex Gaussian distribution with mean vector $\boldsymbol{\mu}$ and covariance matrix $\boldsymbol{\Sigma}$.


\section{System Model}

As shown in Fig.~\ref{system}, we propose a distributed multi-IRS-aided satellite CPWC system, in which the satellite--IRS link is modeled as a Nakagami-$m$ fading channel to represent line-of-sight (LoS) propagation in free space, whereas the IRS--user link is modeled as a Rician fading channel to capture a dominant LoS path superimposed with terrestrial scattering.

In the considered system, a LEO satellite serves as the transmitter and, with the aid of $L$ distributed IRSs, denoted as $I_{1}$, $I_{2}$, \ldots, $I_{l}$, \ldots, $I_{L}$, aims to deliver covert information to an intended user referred to as Bob.
Meanwhile, a warden, denoted by Willie, attempts to detect the presence of the covert communication between the satellite and Bob.
To enable real-time control of the IRSs, a feeder link is established between the gateway and Bob. Bob solely forwards control signals and channel state information (CSI) to the gateway, while all computationally intensive tasks are executed at the gateway. Using the received control information, the gateway configures the phase-shift responses of the IRS elements through dedicated control links.

In the proposed CPWC system, we design a terrestrial-IRS-assisted satellite communication architecture, where distributed IRSs are deployed at fixed terrestrial locations to assist satellite downlink transmission and enhance ground-user coverage~\cite{tang2021wireless}.
Note that the satellite-side IRS architecture is not considered in the model, since space-borne reflecting surfaces introduce stricter payload, operation, calibration, and in-orbit maintenance constraints~\cite{zheng2022intelligent}.
By contrast, ground IRSs can reuse terrestrial infrastructure and be placed near blocked user regions, which provides a practical deployment solution for the CPWC mechanism.

However, the terrestrial-IRS-aided architecture still faces practical challenges, including long propagation distances, significant path loss, and LEO mobility.
Therefore, as a tractable benchmark, the proposed system assumes that the satellite-IRS-user geometry is known and that the channel-aware IRS reconfiguration can be updated within one optimization interval~\cite{Gomez2026RejuvenatingIRS}.

Following the assumptions in~\cite{8930608}, the power of signals undergoing two or more IRS reflections is regarded as negligible owing to severe path loss, and only signals reflected once by the IRSs are retained in the subsequent analysis.

To support the assumptions under the considered system, numerical evaluation shows that the received power contributed by higher-order IRS reflections remains tens of dB below that of the first-order reflected component.
The received power deviation caused by excluding higher-order reflections is only at the level of $10^{-4}~\mathrm{dB}$.
Therefore, higher-order IRS reflections have negligible influence on the total received power under the considered system, while the single-reflection approximation keeps the channel model and covertness analysis tractable.

Additionally, we investigate a practical scenario where the direct links between the satellite and users are inaccessible owing to the heavy shadowing~\cite{huang2018secrecy}.

\begin{figure}[htbp]
\centering
\includegraphics[width=0.35\textwidth]{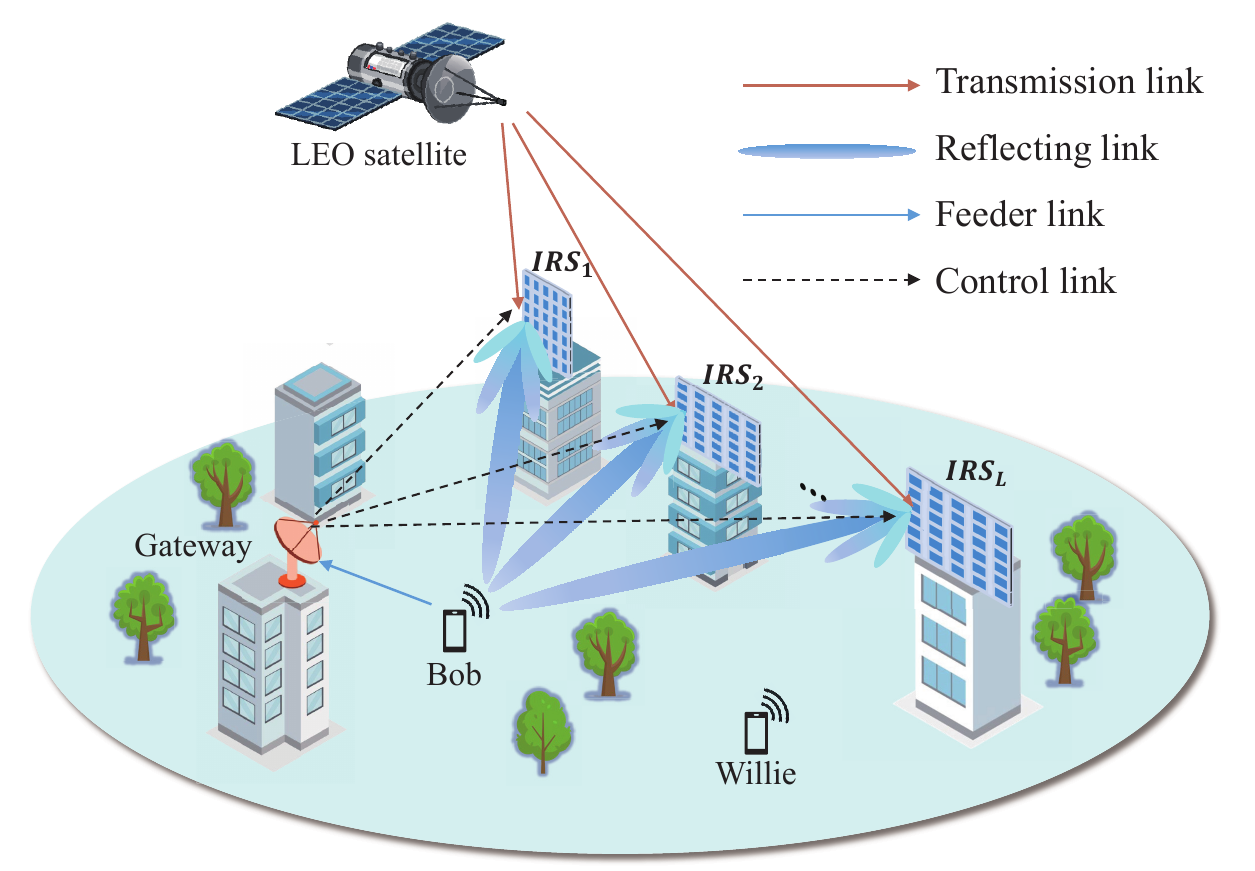}
\caption{Proposed distributed multi-IRS-aided satellite CPWC system, where an LEO satellite, a legitimate user Bob, a warden Willie, $L$ distributed IRSs, and a gateway are included, acting as the covert transmitter, intended receiver, detector, passive reflecting elements, and CSI management node, respectively.}
\label{system}
\end{figure}

As shown in Fig.~\ref{system_3d}, we adopt a three-dimensional (3D) Cartesian coordinate system to specify the positions of all nodes.
Specifically, the coordinates of the satellite, the $l$-th IRS, and each user (Bob or Willie) are denoted by  $(x_s,y_s,z_s)$, $(x_{I_l},y_{I_l},z_{I_l})$, and $(x_u,y_u,z_u),~u\in\left \{b, w \right \}$, respectively.
\begin{figure}[htbp]
\centering
\includegraphics[width=0.35\textwidth]{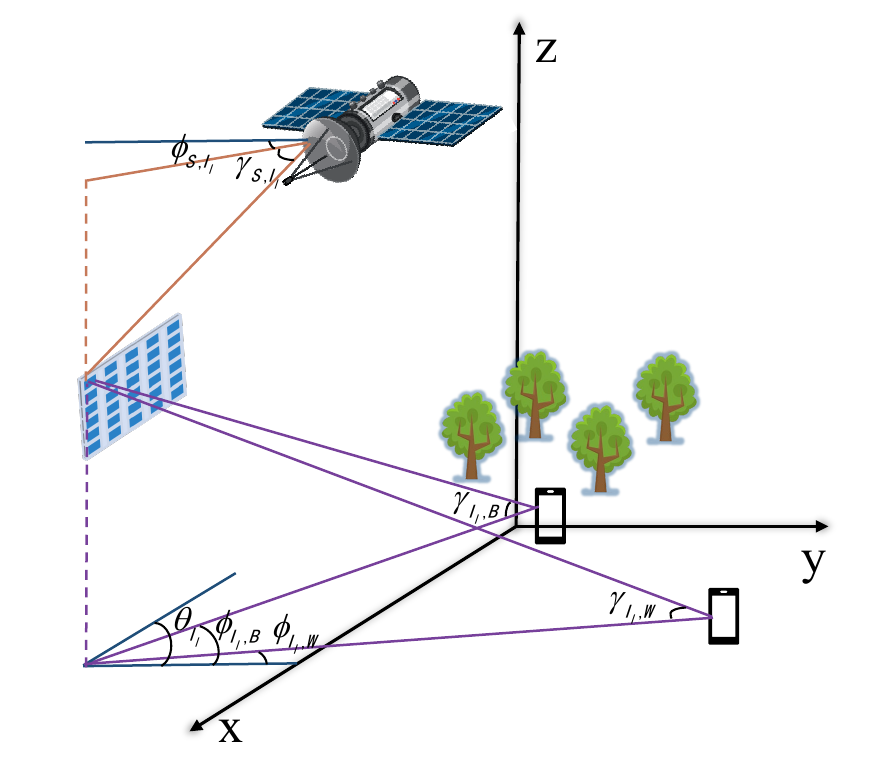}
\caption{Three-dimensional geometric illustration of the satellite–IRS–user links in the proposed CPWC system, taking a single IRS as an example.}
\label{system_3d}
\end{figure}
\subsection{Satellite-IRS Link Channel Model}
To satisfy the stringent precision requirements of SatCom, we employ an RSCS for OFDM in the satellite system.
The set of OFDM subcarriers can be represented as
\begin{align}
\mathcal{A}=\left \{f_ r\mid f_r=f_c+r\Delta f, r=0,1,\cdots,N_s-1\right \},
\end{align}
where $f_c$ is the reference frequency, $\Delta f$ is the subchannel bandwidth, and $B=N_s\Delta f\ll f_c$ ensures that the overall bandwidth remains much smaller than the carrier frequency.
We assume the satellite is equipped with $N$ antennas.
From the set $\mathcal{A}$, we randomly select $N-1$ subcarriers, denoted as $f_n (n=2,3, \cdots, N)$, which serve as the subcarrier frequencies for the satellite's antennas.
The first antenna is fixed at the reference frequency $f_c$, i.e., $f_1 = f_c$.

We assume that the satellite is equipped with a uniform linear array (ULA) of antennas, while each IRS is comprised of $M$ passive reflecting elements arranged in a uniform planar array (UPA).
The wireless channel vector between the satellite and the \(l\)-th IRS is denoted by \(\mathbf{h}_{S,I_l} \in \mathbb{C}^{N \times 1}\), which represents the frequency-dependent channel vector from the \(N\) antennas of the satellite to the reference reflecting element of the \(l\)-th IRS,  and follows the satellite-to-air communication model~\cite{guo2018performance}, given by
\begin{align}
\label{channel1}
\mathbf{h}_{S,I_l}=\lambda\frac{\sqrt{G_S}}{4\pi d_{S,I_l}}\mathbf{z}\odot\mathbf{a}_{S,I_l},
\end{align}
where $\lambda$ is the carrier wavelength, 
$d_{S, I_l}$ is the distance between the satellite and the $l$-th IRS, 
$G_S$ represents the satellite beam gain factor, 
$\mathbf{z}$ following the Nakagami-m distribution represents the shadowing and small-scale fading, and $\mathbf{a}_{S, I_l}$ denotes the array steering vector from the satellite to the reference element of the $l$-th IRS.

Accordingly, the satellite beam gain $G_S$ is given by
\begin{align}
\label{antenna_gain}
G_S=G_S^{\max}\left(\frac{J_1(\nu )}{2\nu }+36\frac{J_3(\nu )}{\nu ^3}\right)^2,
\end{align}
where $G_S^{\max}$ denotes the maximum antenna gain, 
and $J_1(\nu )$ and $J_3(\nu )$ represent the first-order and third-order Bessel functions of the first kind,  respectively. 
In particular, $\nu=2.07123\frac{\sin\gamma_{\mathrm{off}}}{\sin\gamma _{3dB}}$,
where $\gamma_{\mathrm{off}}$ is the off-axis angle for the beam boresight, and $\gamma _{3dB}$ represents the 3 dB beam-width angle of the satellite's antenna.
Additionally, the array steering vector $\mathbf{a}_{S,I_l}$ can be expressed as
\begin{align}
\mathbf{a}_{S,I_l}=\left[e^{j\varLambda_{1}},e^{j\varLambda_{2}},\cdots,e^{j\varLambda_{N}}\right]^T,
\end{align}
where $\varLambda_{n}$ denotes the phase shift from the $n$-th antenna of the satellite to the reference element of the $l$-th IRS. 
In the following, we derive $\varLambda_{n}$.
From Fig.~\ref{system_3d}, we have
\begin{align}
\gamma_{S,{I_l}}=\arctan\frac{h}{\sqrt{\left ( x_s-x_{I_l}\right )^2+\left (y_s-y_{I_l} \right )^2)}}, 
\end{align}
and
\begin{align}
\phi_{S,{I_l}}=\arctan\frac{x_s-x_{I_l}}{y_s-y_{I_l}},
\end{align}
where $h=z_s-z_{I_{l}}$, 
$\gamma_{S, I_l}$ represents the elevation angle of departure (AoD) at the satellite, 
and $\phi_{S, I_l}$ is the azimuth angle of the $l$-th IRS relative to the satellite.
By designating the first element of each IRS as the reference element, the phase of the received signal at the $l$-th IRS from the $n$-th antenna of the satellite can be expressed as
\begin{align}
\omega_{n}=2\pi f_{n}\frac{d_{S,{I_l}}+(n-1)d\cos\gamma  _{S,{I_l}}\cos\phi_{S,{I_l}}}{c},
\end{align}
where $d=\frac{c}{2f_c}$ represents the antenna spacing, with $c$ being the speed of light. 
Accordingly, the relative phase shift from the $n$-th satellite antenna to the reference element of the $l$-th IRS is given by
\begin{align}
\varLambda_{n}=\omega_{n}-\omega_{0},
\end{align}
where $\omega_{0}=2\pi f_{1}\frac{\gamma_{S,{I_l}}}{c}$ corresponds to the phase of the signal received by the $l$-th IRS from the satellite reference antenna.

\subsection{IRS-User Link Channel Model}
The channel between the $l$-th IRS and the users is modeled as a Rician fading incorporating large-scale path loss.
In this case, the wireless channel between the $l$-th IRS and the users is given by
\begin{align}
\mathbf{h}_{I_{l},U}=\sqrt{\mathcal{L}}\mathbf{g},~U\in\left \{B~(\rm{Bob}),\mathit{W}~(\rm{Willie}) \right \},
\end{align}
where $\mathcal{L}$ represents the path loss of the IRS-user link, and $\mathbf{g}$ denotes the small-scale Rician fading vector.

We model the path-loss of the IRS-User link using the path-loss at the reference distance, denoted as $\mathcal{L}_{0}$, and the distance $d_{I_{l}, U}$ between the $l$-th IRS and users, expressed as
\begin{align}
\mathcal{L}=\mathcal{L}_{0}(\frac{d_{I_{l},U}}{d_0})^{-\zeta},
\end{align}
where $d_0=1~\rm m$ is the reference distance and $\zeta$ is the path-loss exponent.

Following \cite{guo2020weighted}, the small-scale fading vector $\mathbf{g}$ is given by
\begin{align}
\mathbf{g}=\sqrt{\frac{K}{K+1}}\mathbf{a}_{I_{l},U}+\sqrt{\frac{1}{K+1}}\widetilde{\mathbf{g}},
\end{align}
where $K$ is the Rician factor, 
$\mathbf{a}_{I_{l}, U}$ and $\widetilde{\mathbf{g}}$ represent the LoS and non-LoS (NLoS) components corresponding to $\mathbf{g}$, respectively.
The NLoS component follows Rayleigh fading, i.e., $\widetilde{\mathbf{g}}\sim\mathbb{CN}(0,\mathbf{I}_M)$, and the LoS component is represented by $\mathbf{a}_{I_{l},U}$.

The Rician IRS-user model provides a general statistical representation, where a Rician factor is used to control the relative strength between the LoS component and the aggregate Rayleigh scattering component. 
Therefore, different urban scattering conditions can be characterized by adjusting the Rician factor.

Note that the phase shift is partly influenced by the frequency $f_n$ of signals received by the IRS from the satellite antenna.
Thus, $\mathbf{a}_{I_{l},U}$ is a function of the subcarrier frequency $f_n$, which can be expressed as
\begin{equation}
\begin{aligned}
\mathbf{a}_{I_{l},U}(f_n)=&[e^{j\varOmega_{1,1}(f_n)},\cdots,e^{j\varOmega_{M_{z},M_{x}}(f_n)}]^{T},
\end{aligned}    
\end{equation}
where $\varOmega_{m_{z},m_{x}}(f_n)$ denote the phase shift introduced by the $(m_{z},m_{x})$-th element of the $l$-th IRS with respect to the reference element. Here, $m_{x}\in {1,2,\cdots, M_{x}}$ and $m_{z}\in {1,2,\cdots, M_{z}}$ represent the horizontal and vertical indices, respectively, where $M_{x}$ and $M_{z}$ denote the number of IRS elements along each dimension, and the total number of reflecting elements is $M=M_{x}M_{z}$. 


The phase shift of the received signal transmitted from the $n$-th satellite antenna to the $(m_{z},m_{x})$-th IRS element, with respect to the reference element of the $l$-th IRS, is expressed as
\begin{equation}
\begin{aligned}
\Delta r_{1_{m_{z},m_{x}}}=&(m_{z}-1)d\sin\gamma_{S,{I_l}}-(m_{x}-1)\\&d\cos\gamma _{S,{I_l}}\cos(\phi _{S,{I_l}}-\theta _{I_l}),  \end{aligned}   
\end{equation}
where $\theta _{I_l}$ represents the placement angle of the $l$-th IRS.

Similarly, the phase shift of the signal received by the user from the $(m_{z},m_{x})$-th IRS element is given by
\begin{equation}
\begin{aligned}
\Delta r_{2_{m_{z},m_{x}}}=&(m_{x}-1)d\cos\gamma _{{I_l},U}\cos(\phi _{{I_l},U}+\theta _{I_l})\\& +(m_{z}-1)d\sin\gamma_{{I_l},U},    
\end{aligned}   
\end{equation}
where the AoD of the $l$-th IRS to user $\gamma_{{I_l},U}$ and the azimuth angle of user relative to the $l$-th IRS $\phi_{{I_l},U}$ can be respectively expressed as
\begin{align} 
\gamma_{{I_l},U}=\arctan\frac{z_{I_{l_1}}}{\sqrt{\left ( x_u-x_{I_l}\right )^2+\left (y_u-y_{I_l} \right )^2)}},
\end{align}
\begin{align}
\phi_{{I_l},U}=\arctan\frac{x_u-x_{I_l}}{y_u-y_{I_l}}.
\end{align}
Thus we have
\begin{equation}
\begin{aligned}
\varOmega_{m_{z},m_{x}}(f_{n}) &= \frac{2\pi f_{n}}{c}\biggl( d_{I_l,U} + \Delta r_{1_{m_{z},m_{x}}} + \Delta r_{2_{m_{z},m_{x}}} \biggr),
\end{aligned}
\end{equation}
where $d_{I_l,U}$ denotes the distance between the $l$-th IRS and users.

Following~\cite{9524501}, the legitimate network is assumed to have perfect instantaneous CSI of the cascaded satellite-IRS-Bob channel.
For the cascaded satellite-IRS-Willie channel, only statistical CSI is available to the legitimate network, and Willie is also assumed to know the corresponding statistical CSI.
The cascaded satellite--IRS--Bob CSI is obtained through pilot training and feedback within each optimization interval.
The satellite transmits predefined pilots over the selected OFDM subcarriers, and the distributed IRSs apply known training reflection patterns.
Bob estimates the effective cascaded channel from the pilot observations and feeds the estimated CSI back to the gateway through the feeder link~\cite{Wang2025ReducingCSIIRS}.
The gateway computes the transmit beamformer and IRS phase shifts, and sends the phase-control commands to the distributed IRSs through dedicated control links.
For the cascaded satellite-IRS-Willie channel, only statistical CSI is available.

Since the fractional bandwidth of each subcarrier is negligible compared with the carrier frequency (i.e., $\Delta f / f_c \ll 1$), the large-scale path gain of the satellite–IRS link in~\eqref{channel1} is approximately frequency-independent. 
Consequently, $f_n$ only affects the phase response in $\mathbf{a}_{S,I_l}$, and the magnitude remains nearly constant over all subcarriers.
\subsection{Signal Transmission}
In the proposed distributed multi-IRS-aided CPWC system, the satellite transmits signal $s(k)$ with zero mean and unit variance to Bob via multiple IRSs.
Accordingly, the received signal at Bob is given by
\begin{align}
\label{y_B}
y_{B}(k)=\sum_{l=1}^{L}\mathbf{h}^{H}_{I_{l},B}\mathbf{\Phi}_{l}\mathbf{H}^{H}_{S,I_{l}}\mathbf{w}s(k)+n_{B}(k),
\end{align}   
where \(k=1,\ldots,K_n\) indexes the transmitted symbols, and $K_n$ is the total number of symbols.
\(\mathbf{\Phi}_l=\operatorname{diag}(\boldsymbol{\theta}_l)\) 
denotes the reflection-coefficient matrix of the \(l\)-th IRS.
The channel vector from the \(l\)-th IRS to Bob is denoted by 
\(\mathbf{h}_{I_{l},B}\), 
\(\mathbf{w}\in\mathbb{C}^{N\times1}\) denotes the beamforming vector for the transmitted signal \(s(k)\), 
and \(n_B(k)\sim\mathcal{CN}(0,\sigma_B^2)\) is the complex Gaussian noise at Bob.

The satellite-to-IRS channel matrix is defined as
\begin{equation}
\mathbf{H}_{S,I_l}
=\big[\mathbf{h}_{S,I_l,1},\mathbf{h}_{S,I_l,2},\ldots,\mathbf{h}_{S,I_l,M}\big],
\end{equation}
where each column $\mathbf{h}_{S,I_l,m}\in\mathbb{C}^{N\times1}$ represents the channel vector
between the satellite’s $N$ transmit antennas and the $m$-th reflecting element of the $l$-th IRS.
The frequency-dependent response of each channel can be further expressed as $\mathbf{h}_{S,I_l,m}$.

The IRS reflection-coefficient vector is modeled as
\begin{align}
\boldsymbol{\theta}_l = \big[ \beta_l^1 e^{j\varphi_l^1}, \ldots, \beta_l^m e^{j\varphi_l^m}, \ldots, \beta_l^M e^{j\varphi_l^M} \big]^T,
\end{align}
where $\beta_l^m\in[0,1]$ and $\varphi_l^m\in[0,2\pi)$, for $m=1,2,\cdots,M$, are the amplitude reflection and phase shift coefficients of the $m$-th element of the $l$-th IRS, respectively.
To maximize reflection power, we set $\beta_l^m=1,~\forall m$.
Accordingly, the achievable covert rate at Bob is expressed as
\begin{align}
R_{B}=\log_{2}\left(1+\frac{\left|\sum_{l=1}^{L}\mathbf{h}_{I_{l},B}^{H}\mathbf{\Phi}_{l}\mathbf{H}^{H}_{S,I_{l}}\mathbf{w}\right|^{2}}{\sigma_{B}^{2}}\right).
\end{align}

The received signal $y_W(k)$ at Willie can be characterized under two competing hypotheses.
Specifically, $\mathcal{H}_0$ denotes the null hypothesis where the satellite remains silent (i.e., no communication with Bob), whereas $\mathcal{H}_1$ corresponds to the alternative hypothesis where covert communication occurs.
In this context, the received signal at Willie can be given by
\begin{equation}
\begin{aligned}
\begin{cases}
\mathcal{H}_{0}:y_{W}(k)=n_{W}(k),\\
\mathcal{H}_{1}:y_{W}(k)=\sum_{l=1}^{L}\mathbf{h}_{I_{l},W}^{H}\mathbf{\Phi}_{l}\mathbf{H}^{H}_{S,I_{l}}\mathbf{w}s(k)+n_{W}(k),
\end{cases}
\end{aligned}    
\end{equation}
where $\mathbf{h}_{I_{l},W}$ denotes the channel vector from the $l$-th IRS to Willie, and $n_W(k) \sim \mathbb{C}\mathbb{N}(0, \sigma_W^2)$ denotes the complex Gaussian noise at Willie.

We assume that Willie adopts a conventional radiometer to detect the covert communication~\cite{7964713}, where the average received signal power $\bar{P}_{W}$ is adopted as the test statistic.
Thus, the decision rule is formulated as
\begin{align}
\bar{P}_{W}=\frac{1}{N_{\mathrm{det}}}\sum_{k=1}^{N_{\mathrm{det}}}|y_W[k]|^2\overset{\mathcal H_{1}}{\underset{\mathcal H_{0}}{\gtrless}}\tau,
\end{align}
where $\tau$ denotes the detection threshold, and $N_{\mathrm{det}}$ denotes the number of received samples used by Willie's radiometer.

The false alarm probability and miss detection probability at Willie are then defined as $P_{\mathrm{FA}} = \Pr\{\mathcal{D}_1 \mid \mathcal{H}_0\}$ and $P_{\mathrm{MD}} = \Pr\{\mathcal{D}_0 \mid \mathcal{H}_1\}$, where $\mathcal{D}_0$ and $\mathcal{D}_1$ denote Willie's decisions in favor of $\mathcal{H}_0$ and $\mathcal{H}_1$, respectively.
The total detection error  probability at Willie is given by
\begin{align}
P_e=P_{FA}+P_{MD}.
\end{align}

\section{Relative-Entropy-Based Warm-Start Design}\label{section_relative_entropy}
\subsection{Detection Performance at Willie}
In this subsection, Willie is assumed to adopt the optimal statistical hypothesis test to improve the detection capability. According to~\cite{6584948}, the minimum detection error probability achieved by the optimal test is given by
\begin{align}
\label{error detection}
P_e^* = 1 - \mathcal{V}_T\left(\mathbb{P}_0,\mathbb{P}_1\right),
\end{align}
where $\mathcal{V}_T\left(\mathbb{P}_0,\mathbb{P}_1\right)$ denotes the total variation distance between $\mathbb{P}_0$ and $\mathbb{P}_1$, which correspond to the probability distributions of the received signal at Willie under $\mathcal{H}_0$ and $\mathcal{H}_1$, respectively, given by
\begin{equation}
\label{probability dis}
\begin{cases}
\mathbb{P}_{0} \triangleq f(y_{W}|\mathcal{H}_{0}) = \mathbb{CN}(0,\sigma_{W}^{2}),\\[2pt]
\mathbb{P}_{1} \triangleq f(y_{W}|\mathcal{H}_{1}) = \\ \qquad \quad \mathbb{CN}\!\left(0,\left|\sum_{l=1}^{L}\mathbf{h}_{I_{l},W}^{H}\mathbf{\Phi}_{l}\mathbf{H}^{H}_{S,I_{l}}\mathbf{w}\right|^{2}+\sigma_{W}^{2}\right).
\end{cases}
\end{equation}
For notational convenience, let
\begin{align}
\sigma_{\nu}^{2}=\left|\sum_{l=1}^{L}\mathbf{h}_{I_{l},W}^{H}\mathbf{\Phi}_{l}\mathbf{H}^{H}_{S,I_{l}}\mathbf{w}\right|^{2}.
\end{align}
Since the exact evaluation of the total variation distance is difficult, Pinsker’s inequality is invoked to obtain the upper bound~\cite{6584948}
\begin{align}
\label{total variation distance}
\mathcal{V}_{T}\left(\mathbb{P}_{0},\mathbb{P}_{1}\right)\leq\sqrt{0.5\,\mathcal{D}\left(\mathbb{P}_{0}\parallel\mathbb{P}_{1}\right)},
\end{align}
where $\mathcal{D}\left(\mathbb{P}_{0}\parallel\mathbb{P}_{1}\right)$ denotes the relative entropy between $\mathbb{P}_{0}$ and $\mathbb{P}_{1}$. By combining~\eqref{error detection} and~\eqref{total variation distance}, the covertness requirement $P_e^*\geq 1-\varepsilon$ with a small positive parameter $\varepsilon$ can be enforced by the relative-entropy-based constraint $\mathcal{D}\left(\mathbb{P}_{0}\parallel\mathbb{P}_{1}\right)\le 2\varepsilon^2$~\cite{6584948}. For the complex Gaussian distributions in~\eqref{probability dis} with variance parameter $\sigma_{\nu}^{2}$, the relative-entropy-based constraint is equivalent to
\begin{align}
\label{covertness constraint}
\ln\left(1+\frac{\sigma_{\nu}^{2}}{\sigma_{W}^{2}}\right)
+\frac{\sigma_{W}^{2}}{\sigma_{\nu}^{2}+\sigma_{W}^{2}}
-1
\le 2\varepsilon^2. 
\end{align}

\subsection{Problem Formulation}
In this subsection, we aim to maximize the covert rate of the proposed CPWC system by jointly optimizing the transmit beamforming vector and the phase shifts of $L$ IRSs, subject to the predefined constraint on covertness.
In this context, the optimization problem is formulated as
\begin{equation}
\label{Problem}
\begin{aligned}
\underset{\mathbf{w},\boldsymbol{\theta}_l}{\max}\quad & R_B \\
\text{s.t.}\quad 
& \|\mathbf{w}\|^2 \leq P_S^{\max},\ \eqref{covertness constraint},\ 
  |\boldsymbol{\theta}_l[m]| = 1,\ \forall l,\forall m,
\end{aligned}
\end{equation}
where $P_S^{max}$ denotes the transmit power budget at the satellite.
Note that $|\mathbf{w}\|^2 \leq P_S^{\max}$ corresponds to the transmit power limitation of the satellite, and \eqref{covertness constraint} ensures the satisfaction of covert communications.

By diagonalizing $\mathbf{h}_{I_l,B}$ as $\mathbf{H}_{I_l,B}=\mathrm{diag}(\mathbf{h}_{I_l,B}^{H})$, the equivalent cascaded channel from the satellite to Bob via the $l$-th IRS is expressed as $\mathbf{G}_{B_l}=\mathbf{H}_{I_l, B}\mathbf{H}^{H}_{S, I_l}$.
Thus, we have $\sum_{l=1}^{L}\mathbf{h}^{H}_{I_{l},B}\mathbf{\Phi}_{l}\mathbf{H}^{H}_{S,I_{l}}\mathbf{w}=\sum_{l=1}^{L}\boldsymbol{\theta}_{l}^{T}\mathbf{G}_{B_l}\mathbf{w}$.
Furthermore, let $\hat{\boldsymbol{\theta}}=\left[\mathbf{\theta}_{1}^{T},\cdots,{\theta}_{l}^{T},\cdots,\mathbf{\theta}_{L}^{T}\right]\in\mathbb{C}^{1\times LM}$ and $\hat{\mathbf{G}}_{B}=\left[\mathbf{G}_{B_1};\cdots;\mathbf{G}_{B_l};\cdots;\mathbf{G}_{B_L}\right]\in\mathbb{C}^{LM\times N}$.
Then, the objective can be transformed into
\begin{align}
\label{simple R}
R_B=\log_{2}{\left ( 1+\frac{\left |\hat{\boldsymbol{\theta}}\hat{\mathbf{G}}_{B}\mathbf{w}\right |^2}{\sigma_B^2} \right )}.
\end{align}

\subsection{Alternative Optimization Algorithm}
The intricate coupling of $\mathbf{w}$ and $\boldsymbol{\theta}_l$ in both the objective function and constraint \eqref{covertness constraint} makes problem \eqref{Problem} difficult to solve.
As a result, an iterative algorithm is developed by leveraging an AO scheme to decouple $\mathbf{w}$ and $\boldsymbol{\theta}_l$. Moreover, constraint \eqref{covertness constraint} is reformulated into an equivalent, more tractable form by modeling its left-hand side as
\begin{align}
f(x)=\ln\left(x\right)+\frac{1}{x}-1,
\end{align}
with $x=1+\frac{\sigma_{\nu}^{2}}{\sigma_W^2} > 1$.
Thus, \eqref{covertness constraint} can be rewritten as $f(x)\le2\varepsilon^2$, where $f(x)$ is an increasing function with respect to $x$, yielding $x\leq f^{-1}(2\varepsilon ^2)$, where $f^{-1}(x)$ represents the inverse function of $f(x)$.
Similar to the derivation of the equivalent channel for Bob, we define the channel parameters for Willie as $\mathbf{H}_{I_l,W}={diag}\big(\mathbf{h}_{I_l,W}^{H}\big)$. The equivalent cascaded channel from the satellite to Willie via the $l$-th IRS is expressed as $\mathbf{G}_{W_l}=\mathbf{H}_{I_l,W}\,\mathbf{H}_{S,I_l}^{H}$. Furthermore, let $\hat{\mathbf{G}}_{W}=[\mathbf{G}_{W_1};\cdots;\mathbf{G}_{W_l};\cdots;\mathbf{G}_{W_L}]$.
Thus, we have $\sigma_{\nu}^{2}=\left|\hat{\boldsymbol{\theta}}\hat{\mathbf{G}}_{W}\mathbf{w}\right |^2$.
Substituting $\sigma_{\nu}^{2}$ into $x\leq f^{-1}(2\varepsilon ^2)$ yields 
\begin{align}
\label{simple problemb}
1+\frac{\left |\hat{\boldsymbol{\theta}}\hat{\mathbf{G}}_{W}\mathbf{w}\right |^2}{\sigma_{W}^2}\leq f^{-1}(2\varepsilon ^2).
\end{align}
With the above derivations, the optimization problem is decomposed into two sub-problems, which can be solved alternately to obtain the optimal transmit and reflection beamformers, respectively.

\emph{1)~Sub-problem 1: Transmit beamforming optimization}
From \eqref{simple R}, $R_B$ monotonically increases with Bob's SNR $\gamma_B$, and thus the objective function can be reformulated in terms of $\gamma_B$.
Given $\hat{\mathbf{\theta}}$, we first formulate the sub-problem 1 as
\begin{equation}
\label{Problem1.1}
\begin{aligned}
\underset{\mathbf{w}}{\text{max}}&\quad~\gamma_B=\frac{\left |\hat{\boldsymbol{\theta}}\hat{\mathbf{G}}_{B}\mathbf{w}\right |^2}{\sigma_B^2} \\
\text {s.t.}&\quad\|\mathbf{w}\|^2 \leq P_S^{\max}, \eqref{simple problemb}.
\end{aligned}
\end{equation}
However, \eqref{Problem1.1} remains non-convex since the objective function is not concave with respect to $\mathbf{w}$.
To address this issue, we introduce a new variable $\mathbf{W}=\mathbf{w}\mathbf{w}^H\in\mathbb{C}^{N\times N}$, and define
\begin{align}
\mathbf{C_B}=\frac{{\hat{\mathbf{G}}_{B}^H}\hat{\boldsymbol{\theta}}^H\hat{\boldsymbol{\theta}}\hat{\mathbf{G}}_{B}}{\sigma_B^2},    
\mathbf{C_W}=\frac{{\hat{\mathbf{G}}_{W}^H}\hat{\boldsymbol{\theta}}^H\hat{\boldsymbol{\theta}}\hat{\mathbf{G}}_{W}}{\sigma_W^2}.  
\end{align}
Therefore, the sub-problem can be equivalently transformed into 
\begin{align}
\underset{\mathbf{W}}{\text{max}}\quad&~\text{tr}(\mathbf{C_B}\mathbf{W})\label{Problem1.1.0}\\
\text {s.t.}\quad&~\text{tr}(\mathbf{W})\leq P_S^{max},\tag{\ref{Problem1.1.0}{a}} \label{Problem1.1.0a}\\
&~\text{tr}(\mathbf{C_W}\mathbf{W})\leq f^{-1}(2\varepsilon ^2)-1,\tag{\ref{Problem1.1.0}{b}}\label{Problem1.1.0b}\\
&~\mathbf{W}\succeq \mathbf{0},\tag{\ref{Problem1.1.0}{c}}\label{Problem1.1.0c}\\
&~\text{rank}(\mathbf{W})=1.\tag{\ref{Problem1.1.0}{d}}\label{Problem1.1.0d}
\end{align}
Since $\mathbf{W}$ is a rank-one symmetric positive semidefinite (PSD) matrix, the SDR method can be applied to drop the rank-one constraint.
As a result, \eqref{Problem1.1.0} is converted to
\begin{align}
\underset{\mathbf{W}}{\text{max}}\quad&~\text{tr}(\mathbf{C_B}\mathbf{W})\label{Problem1.1.1}\\
\text {s.t.}\quad&~\eqref{Problem1.1.0a},\eqref{Problem1.1.0b},\eqref{Problem1.1.0c}.\tag{\ref{Problem1.1.1}{b}}\label{Problem1.1.1b}
\end{align}
\eqref{Problem1.1.1} is a convex semidefinite programming (SDP) problem that can be solved by CVX which satisfies the rank-one constraint.
Thus, dropping constraint \eqref{Problem1.1.0d} does not compromise the optimality of the original problem \eqref{Problem1.1.0}.
\begin{thm}\label{thm1}
The optimal solution $\mathbf{W}^*$ to problem \eqref{Problem1.1.1} satisfies the rank-one property, i.e., $\text{rank}(\mathbf{W}^*)=1$.
\end{thm}
\begin{IEEEproof}
See Appendix~\ref{proof of thm1}.
\end{IEEEproof}

To retrieve the optimal transmit beamforming vector $\mathbf{w}^*$ from $\mathbf{W}^*$,we apply eigenvalue decomposition (EVD) on $\mathbf{W}^*$, yielding $\mathbf{W}^*=\lambda_{\mathbf{W}^*_{\text{max}}}\mathbf{x}\mathbf{x}^H$, where $\lambda_{\mathbf{W}^*_{\text{max}}}$ is the maximum eigenvalue of $\mathbf{W}^*$, and $\mathbf{x}$ is the corresponding normalized eigenvector.
Consequently, the optimal beamforming vector is obtained as $\mathbf{w}^*=\sqrt{\lambda_{\mathbf{W}^*_{\text{max}}}}\mathbf{x}$.

\emph{2)~Sub-problem 2: Reflection beamforming optimization}

In the subsequent step, we optimize the phase shifts of $L$ IRSs for a given $\mathbf{w}$.
Jointly considering~\eqref{simple problemb} and~\eqref{Problem}, the optimization problem is reformulated as
\begin{equation}
\label{Problem1.2}
\begin{aligned}
   \underset{\hat{\boldsymbol{\theta}}}{\text{max}}\quad&~\gamma_B=\frac{\left |\hat{\boldsymbol{\theta}}\hat{\mathbf{G}}_{B}\mathbf{w}\right |^2}{\sigma_B^2} \\
\text {s.t.}\quad&|\boldsymbol{\theta}_l[m]| = 1,\ \forall l,\forall m, \eqref{simple problemb}. 
\end{aligned}
\end{equation}
The problem \eqref{Problem1.2} is non-convex due to its non-concave objective function and unit-modulus constraint on the IRS reflection coefficients.
To address this challenge, we simplify the objective function and~ \eqref{simple problemb} with auxiliary variables $\mathbf{D_B}$ and $\mathbf{D_W}$, given by
\begin{align}
\mathbf{D_B}=\frac{{\hat{\mathbf{G}}_{B}}{\mathbf{w}}{\mathbf{w}^H}\hat{\mathbf{G}}_{B}^H}{\sigma_B^2}, \quad   
\mathbf{D_W}=\frac{{\hat{\mathbf{G}}_{w}}{\mathbf{w}}{\mathbf{w}^H}\hat{\mathbf{G}}_{W}^H}{\sigma_W^2}.
\end{align}

Accordingly, the objective function can be reformulated as $\text{tr} ( \hat{\boldsymbol{\theta}}\mathbf{D_B}\hat{\boldsymbol{\theta}}^H )$.
Similarly,~\eqref{simple problemb} is equivalent to $\text{tr} ( \hat{\boldsymbol{\theta}}\mathbf{D_W}\hat{\boldsymbol{\theta}}^H )\leq f^{-1}(2\varepsilon ^2)-1$.
Second, we define $\mathbf{\Theta}=\hat{\boldsymbol{\theta}}^H\hat{\boldsymbol{\theta}}$, thus the sub-problem 2 can be transformed into
\begin{equation}
  \begin{aligned}
\underset{\mathbf{\Theta}}{\text{max}}\quad&~\text{tr}(\mathbf{D_B}\mathbf{\Theta})\label{Problem1.2.1}\\
\text {s.t.}\quad&~\text{tr}(\mathbf{D_W}\mathbf{\Theta})\leq f^{-1}(2\varepsilon ^2)-1,\\
&~\mathbf{\Theta}\succeq \mathbf{0},~\text{rank}(\mathbf{\Theta})=1,\\
&~\mathbf{\Theta}_{n,n}=1,~n=1,2,\cdots ,lM.
\end{aligned}  
\end{equation}
Since the rank-one constraint renders the problem non-convex, it is relaxed using the SDR method.
Let $\mathbf{\Theta}^*$ denote the optimal solution to~\eqref{Problem1.2.1}.
However, $\mathbf{\Theta}^*$ is not guaranteed to be a rank-one matrix.
Thus, if $\text{rank}(\mathbf{\Theta}^*>1)$, a rank-one solution can be reconstructed by using the Gaussian randomization method. 

Specifically, we first perform EVD on $\mathbf{\Theta}^*$, yielding $\mathbf{\Theta}^*=\mathbf{U}\mathbf{\Sigma} \mathbf{U}^H$, where $\mathbf{U}$ is the eigenvector matrix corresponding to $\mathbf{\Theta}^*$, and $\mathbf{\Sigma}$ is a diagonal matrix with each diagonal element representing an eigenvalue of $\mathbf{\Theta}^*$.
Thus, we obtain $\bar{\boldsymbol{\theta}^*}=\mathbf{U}\sqrt{\mathbf{\Sigma}} \mathbf{r}$, where $\mathbf{r} \sim \mathcal{C}\mathcal{N}(0,\mathbf{I_N})$ is a Gaussian randomization vector.
Among the obtained candidates, the one maximizing the objective function in~\eqref{Problem1.2.1} is selected as the optimal reflection beamforming vector $\boldsymbol{\theta}$. 
To satisfy the modulus-1 constraint, the final solution to problem~\eqref{Problem1.2.1} is obtained as $\boldsymbol{\theta}^*=e^{j\arg (\bar{\mathbf{\theta}^*})}$.

In practice, realizing continuous (infinite-resolution) phase shifts at the IRS is challenging due to hardware limitations~\cite{9122596}.
Thus, we further consider IRSs with finite phase-shift resolution and optimize discrete phase shifts.
Assuming 1-bit reflection levels, i.e., possible phase shift $\theta\in \left \{ -1,1\right \}$, we employ the nearest point projection (NPP)~\cite{guo2019weightedsumrateoptimizationintelligent} to determine the discrete phase shifts for IRSs.
Specifically, the continuous phase shifts obtained from problem~\eqref{Problem1.2} are projected onto the nearest feasible points in the set $\{-1,1\}$.
\begin{algorithm}[t]
\caption{AO--SDR Algorithm for Multi-IRS-Aided CPWC System}
\label{alg1}
\begin{algorithmic}[1]
\STATE Initialize $\hat{\boldsymbol{\theta}}^{0}$, set $\delta$, $I_{\max}$, $i=1$, $v=1$.
\WHILE{$v>\delta$ and $i<I_{\max}$}
    \STATE Solve~\eqref{Problem1.1.1} with $\hat{\boldsymbol{\theta}}^{i-1}$ to get $\mathbf{W}^i$, extract $\mathbf{w}^i$ via EVD.
    \STATE Solve SDR problem~\eqref{Problem1.2.1} with $\mathbf{w}^i$ to get $\mathbf{\Theta}^i$.
    \IF{$\mathrm{rank}(\mathbf{\Theta}^i)=1$}
        \STATE Set $\boldsymbol{\theta}^i$ as principal eigenvector.
    \ENDIF
    \STATE Update $R_{B}^i$, $v=\frac{R_{B}^i-R_{B}^{i-1}}{R_{B}^i}$, $i\!\leftarrow\!i+1$.
\ENDWHILE
\IF{$\mathrm{rank}(\mathbf{\Theta}^i)>1$}
    \STATE Apply Gaussian randomization on $\mathbf{\Theta}^i$ to get $\{\hat{\boldsymbol{\theta}}^{(k)}\}$, select $\hat{\boldsymbol{\theta}}^{\star}\!=\!\arg\max_{k}R_B^{(k)}$.
\ELSE
    \STATE $\hat{\boldsymbol{\theta}}^{\star}\!=\!\boldsymbol{\theta}^i$.
\ENDIF
\STATE \textbf{Output:} $\mathbf{w}^{\star}=\mathbf{w}^i$, $\hat{\boldsymbol{\theta}}^{\star}$, $R_B^{\star}$.
\end{algorithmic}
\end{algorithm}

\emph{3)~Proposed Optimization Algorithm Analysis}

The proposed AO-based algorithm for problem~\eqref{Problem} in the multi-IRS-aided CPWC system is summarized in Algorithm~\ref{alg1}. 
It alternately solves two subproblems in the SDR-domain until the objective improvement falls below a threshold $\delta$, and then applies Gaussian randomization to obtain a feasible rank-one solution.
Convergence is guaranteed for two reasons. 
First, monotonic convergence holds because the AO iterations are performed entirely in the SDR domain. After convergence, Gaussian randomization is applied once as a post-processing step to obtain a feasible rank-one solution. 
Second, the iterative process converges as the objective value of problem~\ref{Problem} is upper bounded by the covertness and power constraint.

The computational complexity of Algorithm~1 mainly comes from solving the two SDR-based subproblems in each AO iteration. 
Since Subproblem~1 optimizes the lifted transmit beamforming matrix with dimension $N$, its SDR complexity is $\mathcal{O}(N^{3.5})$. Subproblem~2 optimizes the lifted IRS phase-shift matrix with dimension $LM$, and its SDR complexity is $\mathcal{O}((LM)^{3.5})$. 
Let $I_{\mathrm{total1}}$ denote the total number of AO iterations. 
When $\boldsymbol{\Theta}$ is not rank-one, Gaussian randomization with $N_{\mathrm{rand}}$ trials is used for rank-one recovery, which introduces an additional post-processing complexity of $\mathcal{O}(N_{\mathrm{rand}}(LM)^2)$. 
Therefore, the overall computational complexity of Algorithm~1 is 
\begin{align}
\mathcal{O}\!\left(
I_{\mathrm{total1}}
\left[
N^{3.5}+(LM)^{3.5}
\right]
+
N_{\mathrm{rand}}(LM)^2
\right).    
\end{align}
~It is worth highlighting that Algorithm~\ref{alg1} is not intended to be the final optimal design. Instead, it provides a feasible high-quality solution that serves as a warm start for subsequent refinements in the next section.

\subsection{Joint Optimization Algorithm}
In this subsection, a joint optimization algorithm is developed to improve performance beyond Algorithm~\ref{alg1}. 
Although the AO method alternately updates the transmit beamformer and IRS phase shifts, it may converge to a suboptimal stationary point due to the decoupled structure. A successive convex approximation (SCA)-based second-order cone programming (SOCP) framework~\cite{9961865} is adopted to address the limitation, which jointly optimizes the transmit beamforming and IRS phase shifts in each iteration to enhance convergence efficiency and achievable performance.

For clarity, we briefly outline the SOCP scheme.
Let $\mathbf{g}_{1}=\frac{\hat{\mathbf{\theta}}\hat{\mathbf{G}}_{B}}{\sigma_B}$, and thus we have $\gamma_B=|\mathbf{g}_{1}\mathbf{w}|^2$.
According to~\cite{9961865}, the following inequality holds
\begin{equation}\label{g_1w}
\begin{aligned}
|\mathbf{g}_{1}\mathbf{w}|^{2}
&\geq
\Re\big\{\big(\mathbf{b}^{(i)}\big)^{H}
        \big[\big(a^{(i)}\big)\mathbf{g}_{1}^{H}+\mathbf{w}\big]\big\}
        - \frac{1}{2}\big\|\mathbf{b}^{(i)}\big\|^{2} \\
&\quad
  - \frac{1}{2}\Big\|\big(a^{(i)}\big)\mathbf{g}_{1}^{H}-\mathbf{w}\Big\|^{2}
  - \big|a^{(i)}\big|^{2} \triangleq \varpi,
\end{aligned}
\end{equation}
with
\begin{align}
a^{(i)}=\mathbf{g}_{1}^{(i)}\mathbf{w}^{(i)},   
\end{align}
and
\begin{align}
\mathbf{b}^{(i)}=a^{(i)}(\mathbf{g}_{1}^{(i)})^{H}+\mathbf{w}^{(i)},
\end{align}
where $a^{(i)}$ and $\mathbf{b}^{(i)}$ denote the values at the $i$-th iteration in the optimization process.
Similarly, the left hand of~\eqref{simple problemb} can be rewritten as $1+\left |\mathbf{g}_{2}\mathbf{w}\right |^2$, where $\mathbf{g}_{2}=\frac{ |\hat{\mathbf{\theta}}\hat{\mathbf{G}}_{W}|}{\sigma_{W}}$.
To facilitate the reformulation of the problem, we further employ two slack variables that satisfy 
\begin{align}
\label{slack1}
t\geq|\Re\{\mathbf{g}_{2}\mathbf{w}\}|,
\end{align}\begin{align}
\label{slack2}
\bar{t}\geq\:|\Im\{\mathbf{g}_{2}\mathbf{w}\}|.
\end{align}
Since $\left |\mathbf{g}_{2}\mathbf{w}\right |^2=|\Re\{\mathbf{g}_{2}\mathbf{w}\}|^2+|\Im\{\mathbf{g}_{2}\mathbf{w}\}|^2\leq t^2+\bar{t}^2$,~\eqref{simple problemb} can be given by
\begin{align}
\label{slack_cons}
t^2+\bar{t}^2\leq f^{-1}(2\varepsilon ^2)-1.   
\end{align}
Then, we can convert~\eqref{slack1} into two constraints~\eqref{slack11}~and~\eqref{slack12}.
\begin{figure*}[!t]
\normalsize
\begin{align}
t &\geq \frac{1}{4}\left\|\mathbf{g}_{2}^{H}+\mathbf{w}\right\|^{2}
      + \frac{1}{4}\left\|\left(\mathbf{g}_{2}^{(n)}\right)^{H}-\mathbf{w}^{(n)}\right\|^{2}
      - \frac{1}{2}\Re\left\{\left(\mathbf{g}_{2}^{(n)}-\left(\mathbf{w}^{(n)}\right)^{H}\right)
      \left(\mathbf{g}_{2}^{H}-\mathbf{w}\right)\right\}
      \label{slack11} \\
t &\geq \frac{1}{4}\left\|\mathbf{g}_{2}^{H}-\mathbf{w}\right\|^{2}
      + \frac{1}{4}\left\|\left(\mathbf{g}_{2}^{(n)}\right)^{H}+\mathbf{w}^{(n)}\right\|^{2}
      - \frac{1}{2}\Re\left\{\left(\mathbf{g}_{2}^{(n)}+\left(\mathbf{w}^{(n)}\right)^{H}\right)
      \left(\mathbf{g}_{2}^{H}+\mathbf{w}\right)\right\}
      \label{slack12} \\
\bar{t} &\geq \frac{1}{4}\left\|\mathbf{g}_{2}^{H}-j\mathbf{w}\right\|^{2}
          + \frac{1}{4}\left\|\left(\mathbf{g}_{2}^{(n)}\right)^{H}+j\mathbf{w}^{(n)}\right\|^{2}
          - \frac{1}{2}\Re\left\{\left(\mathbf{g}_{2}^{(n)}-j\left(\mathbf{w}^{(n)}\right)^{H}\right)
          \left(\mathbf{g}_{2}^{H}+j\mathbf{w}\right)\right\}
          \label{slack21} \\
\bar{t} &\geq \frac{1}{4}\left\|\mathbf{g}_{2}^{H}+j\mathbf{w}\right\|^{2}
          + \frac{1}{4}\left\|\left(\mathbf{g}_{2}^{(n)}\right)^{H}-j\mathbf{w}^{(n)}\right\|^{2}
          - \frac{1}{2}\Re\left\{\left(\mathbf{g}_{2}^{(n)}+j\left(\mathbf{w}^{(n)}\right)^{H}\right)
          \left(\mathbf{g}_{2}^{H}-j\mathbf{w}\right)\right\}
          \label{slack22}
\end{align}

\end{figure*}
Similarly,~\eqref{slack2} can be replaced by~\eqref{slack21}~and~\eqref{slack22}.

Therefore, we derived an equivalent optimization problem for~\eqref{Problem}, given by
\begin{equation}
    \begin{aligned}
\underset{\mathbf{w},\mathbf{\theta}_l}{\text{max}}\quad&~\varpi+\varsigma \left [ 2\Re\left \{ \mathbf{\theta}_l^{(i)}\mathbf{\theta}_l^H\right \}-\left \|{\mathbf{\theta}_l^{(i)}}^H \right \|^2\right ]\label{socp}\\
\text {s.t.}\quad&\|\mathbf{w}\|^2 \leq P_S^{\max},|\mathbf{\theta}_l[m]|\leq1,\forall l,\forall m,\\
&~\eqref{slack_cons},~\eqref{slack11},~\eqref{slack12},~\eqref{slack21},~\eqref{slack22},\\
\end{aligned}
\end{equation}
where $\varsigma>0$ is the regularization parameter. 
Under the SCA framework, the optimal transmit beamforming vector and IRS phase shifts 
are obtained by iteratively solving the convex problem in~\eqref{socp}.

\section{Covert-oriented scheme based on detection error probability with warm-up initialization}
In this section, we further tighten the covertness modeling by constraining the detection error probability. In particular, we adopt a warm-start strategy, where the beamformer obtained from Algorithm~\ref{alg1} is used to initialize an SQP-based algorithm.
\subsection{Detection Performance at Willie}
In this subsection, we evaluate the performance of Willie with the detection error probability on the covert communications.
According to~\cite{8624357}, the false alarm and miss detection probabilities can be respectively derived as
\begin{align}
P_{\mathrm{FA}}
=
\mathrm{Pr}(\bar{P}_{W}>\tau|\mathcal{H}_{0})
=
1-
\frac{
\Gamma_{\mathrm{l}}\left(N_{\mathrm{det}},\frac{N_{\mathrm{det}}\tau}{\sigma_{W}^{2}}\right)
}
{\Gamma(N_{\mathrm{det}})},
\end{align}
and
\begin{align}
P_{\mathrm{MD}}
=
\mathrm{Pr}(\bar{P}_{W}<\tau|\mathcal{H}_{1})
=
\mathbb{E}_{Y}
\left[
\frac{
\Gamma_{\mathrm{l}}\left(N_{\mathrm{det}},\frac{N_{\mathrm{det}}\tau}{\sigma_{W}^{2}+Y}\right)
}
{\Gamma(N_{\mathrm{det}})}
\right],
\end{align}
where $\Gamma_{\mathrm{l}}(a,x)=\int_{0}^{x}e^{-t}t^{a-1}{\rm d}t$ is the lower incomplete Gamma function, and $\Gamma(a)$ is the complete Gamma function, 
and $Y=\left |\sum_{l=1}^{L}\mathbf{h}_{I_{l},W}^{H}\mathbf{\Phi}_{l}\mathbf{H}^{H}_{S,I_{l}}\mathbf{w}\right |^2$.
The probability density function (PDF) of $Y(\theta)$ is given by Theorem~\ref{thm2}.

\begin{thm}\label{thm2}
Given that the true distribution of $Y$ can be approximated by the Gamma distribution
\begin{equation*}
Y(\theta)\dot{\sim}\text{Gamma}\bigl(\alpha(\theta),\beta(\theta)\bigr),
\end{equation*}
and the PDF of $Y$ is expressed as
\begin{align}\label{f_Y}
f_Y(y;\alpha,\beta)=\frac{y^{\alpha-1} }{\beta ^\alpha\Gamma (\alpha )}e^{-\frac{y}{\beta}},
\end{align}
where $\alpha>0$ is the shape parameter and $\beta>0$ is the scale parameter of the distribution.
\end{thm}
\begin{IEEEproof}
See Appendix~\ref{proof of thm2}.
\end{IEEEproof}
Note that $\alpha$ and $\beta$ can be calculated using~\eqref{alpha} and~\eqref{beta}.

In this context, we introduce
\begin{align*}
\mathbf{c}_{S,I_l}
&=\lambda \frac{\sqrt{G_S}}{4\pi d_{S,I_l}}\mathbf{a}_{S,I_l},\\
\mathbf{J}_{I_l,W}
&=\mathrm{diag}(\mathbf{c}_{S,I_l})\mathbf{H}_{I_l,W}^{(f)},
\end{align*}
where
\begin{align*}
\mathbf{H}_{I_l,W}^{(f)} \triangleq
\big[\mathbf{h}_{I_l,W}(f_1),\,\mathbf{h}_{I_l,W}(f_2),\,\ldots,\,\mathbf{h}_{I_l,W}(f_N)\big].
\end{align*}
Thus it holds that
\begin{align}
\mathcal{K}_i=|\mathbf{w}\odot(\mathbf{J}_{I_l,W}\theta_l)|.   
\end{align}
To further simplify calculations, $\hat{\mathbf{J}_{W}}$ is set as
\begin{align}
\begin{bmatrix}
    \mathbf{J}_{I_1,W} &        & 0 \\
        & \ddots &   \\
    0   &        & \mathbf{J}_{I_l,W}\\
\end{bmatrix},
\end{align}
and then $\mathcal{K}_i$ can be calculated by
\begin{align}
\mathcal{K}_i=
\left |\begin{bmatrix}
\mathbf{w}\\
\mathbf{w}
\end{bmatrix}\odot(\hat{\mathbf{J}_{W}}\hat{\mathbf{\theta}}^T)\right |_i.
\end{align}
Accordingly, the optimal detection threshold can be determined by solving $\frac{\partial P_e}{\partial \tau }=0$, which yields
\begin{equation}{\label{detection_threshold_temp}}
    \begin{aligned}
\frac{e^{-N_{\mathrm{det}}\tau^*/\sigma_W^2}}{\sigma_W^{2N_{\mathrm{det}}}}
=
\mathbb{E}_Y
\left[
\frac{e^{-N_{\mathrm{det}}\tau^*/(\sigma_W^2+Y)}}{(\sigma_W^2+Y)^{N_{\mathrm{det}}}}
\right], 
\end{aligned}
\end{equation}
where $\mathbb{E}_Y[\cdot]$ denotes the expectation with respect to $Y$.

However, (\ref{detection_threshold_temp}) is analytically intractable. To address this challenge, we utilize the Gaussian-Chebyshev quadrature to express the probability of error detection as~\eqref{Pe_GC}.
\begin{figure*}[ht]
	\begin{equation}\label{Pe_GC}
        \begin{aligned}
		P_{e}
        &=\:1-\frac1{\Gamma(N_{\mathrm{det}})}
        \left[
        \Gamma_{\mathrm{l}}\left(N_{\mathrm{det}},\frac{N_{\mathrm{det}}\tau}{\sigma_W^2}\right)
        -\mathbb{E}_{Y}
        \left[
        \Gamma_{\mathrm{l}}\left(N_{\mathrm{det}},\frac{N_{\mathrm{det}}\tau}{Y+\sigma_W^2}\right)
        \right]
        \right]\\
        &\overset{(a)}{\operatorname*{=}}
        1-\frac1{\Gamma(N_{\mathrm{det}})}
        \left[
        \Gamma_{\mathrm{l}}\left(N_{\mathrm{det}},\frac{N_{\mathrm{det}}\tau}{\sigma_W^2}\right)
        -\int_0^{\frac\pi2}
        \Gamma_{\mathrm{l}}\left(N_{\mathrm{det}},\frac{N_{\mathrm{det}}\tau}{\tan\vartheta+\sigma_W^2}\right)
        \frac{f_{Y}(\tan\vartheta)}{\cos^2\vartheta}\rm d \vartheta
        \right]\\
        &\overset{(b)}{\operatorname*{=}}
        1-\frac1{\Gamma(N_{\mathrm{det}})}
        \left[
        \Gamma_{\mathrm{l}}\left(N_{\mathrm{det}},\frac{N_{\mathrm{det}}\tau}{\sigma_W^2}\right)
        -\frac{\pi}{N_\text{quad}}\sum_{i=1}^B
        \Gamma_{\mathrm{l}}\left(N_{\mathrm{det}},\frac{N_{\mathrm{det}}\tau}{\tan\vartheta_i+\sigma_W^2}\right)
        \times
        \frac{f_{Y}(\tan\vartheta_i)\sqrt{\vartheta_i\left(\frac\pi2-\vartheta_i\right)}}{\cos^2\vartheta_i}
        \right].
        \end{aligned}
	\end{equation}
        \hrulefill
\end{figure*}
In~\eqref{Pe_GC}, step $(a)$ corresponds to the variable substitution $Y=\tan\vartheta$, while step $(b)$ applies the Gauss-Chebyshev quadrature, where $N_\text{quad}$ represents the number of summation terms controlling the complexity of precision.
The parameter of the quadrature is given by $\vartheta_i=\frac\pi4(1+\cos\frac{(2i-1)\pi}{2N_\text{quad}})$.
Subsequently, by setting the first-order derivative of $P_{e}$ with respect to the detection threshold to zero, the optimal detection threshold can be obtained by
\begin{equation}\label{tau*}
\begin{aligned}
    &\tau^*=-\frac{\sigma_W^2}{N_{\mathrm{det}}} \times
    \\&\ln\left[
    \frac{\pi}{N_\text{quad}}
    \sum_{i=1}^B
    \frac{
    f_{Y}(\tan\vartheta_i)
    \sqrt{\vartheta_i\left(\frac{\pi}{2}-\vartheta_i\right)}
    e^{-\frac{N_{\mathrm{det}}\tau^*}{\tan\vartheta_i+\sigma_W^2}}
    }
    {
    \left(\tan\vartheta_i/\sigma_W^2+1\right)^{N_{\mathrm{det}}}
    \mathrm{cos}^2\vartheta_i
    }
    \right].
\end{aligned}
\end{equation}
Substituting~\eqref{tau*} into~\eqref{Pe_GC}, the optimal detection error probability can be expressed as~\eqref{errorp}.

\begin{figure*}
\begin{equation}\label{errorp}
    \begin{aligned}
     P_{e}^*
     &=1-\frac1{\Gamma(N_{\mathrm{det}})}
     \left[
     \Gamma_{\mathrm{l}}\left(N_{\mathrm{det}},\frac{N_{\mathrm{det}}\tau^*}{\sigma_W^2}\right)
     -\frac{\pi}{N_\text{quad}}\sum_{i=1}^B
     \Gamma_{\mathrm{l}}\left(N_{\mathrm{det}},\frac{N_{\mathrm{det}}\tau^*}{\tan\vartheta_i+\sigma_W^2}\right)
     \times
     \frac{f_{Y}(\tan\vartheta_i)\sqrt{\vartheta_i\left(\frac\pi2-\vartheta_i\right)}}{\cos^2\vartheta_i}
     \right].
    \end{aligned}
\end{equation}
\hrulefill
\end{figure*}

\subsection{Optimization Problem Formulation}
Based on the above analysis, the optimization problem can be formulated as 
\begin{equation}
   \begin{aligned}
\underset{\mathbf{w}}{\text{max}}\quad&~R_B\big(\mathbf{w}, \mathbf{\Theta}(\mathbf{w})\big)\label{ao_error} \\
\text {s.t.}\quad&~\|\mathbf{w}\|^2\leq P_S^{max}, {P_{e}}^*(\mathbf{w})\geq 1-\varepsilon.
\end{aligned} 
\end{equation}
Due to the imposed cascaded phase-alignment rule, the IRS phase-shift matrix $\mathbf{\Theta}(\mathbf{w})$ is uniquely determined by the transmit beamformer and inherently satisfies the unit-modulus property $|\theta_l[m]| = 1, \forall l,m$, so that the detection error probability $P_e^*$ can be regarded as a function of $\mathbf{w}$ only.

\subsection{Algorithm Design}
Under the cascaded phase-alignment rule, the phase-aligned IRS reflection vector associated with a given transmit beamformer $\mathbf{w}$ is given by
\begin{align}
\hat{\boldsymbol{\theta}}^{\star}(\mathbf{w})
= \exp\!\big(-j\,\arg\{\hat{\mathbf{G}}_B \mathbf{w}\}\big),
\label{theta_mapping}
\end{align}
where the exponential and argument operators are applied element-wise. Equivalently, the phase of the $m$-th reflecting element of the $l$-th IRS is
\begin{align}
\theta_{l}^{\star}[m](\mathbf{w})
= \exp\!\big(-j\,\arg\{[\mathbf{G}_{Bl}\mathbf{w}]_{m}\}\big),
~\forall\,l,m,
\label{theta_mapping_element}
\end{align}
where $\mathbf{G}_{Bl}$ denotes the cascaded channel corresponding to the $l$-th IRS and $[\cdot]_m$ extracts its $m$-th entry.
With \eqref{theta_mapping}--\eqref{theta_mapping_element}, the IRS phase-shift matrix can be written as $\mathbf{\Theta}(\mathbf{w}) = \mathbf{\Theta}\big(\hat{\boldsymbol{\theta}}^{\star}(\mathbf{w})\big)$.

Substituting the deterministic mapping $\mathbf{\Theta}(\mathbf{w})$ into the covert rate and the covertness constraint, the design problem in \eqref{ao_error} reduces to a single-variable nonlinear program with respect to the transmit beamformer $\mathbf{w}$. Since $R_B\big(\mathbf{w}, \mathbf{\Theta}(\mathbf{w})\big)$ is monotonically increasing in Bob's receive SNR, \eqref{ao_error} can be equivalently written in the SNR-maximization form as
\begin{equation}
\label{snr_opt}
\begin{aligned}
\underset{\mathbf{w}}{\text{max}}\quad
& \gamma_B(\mathbf{w}) \triangleq
\frac{\big|\hat{\boldsymbol{\theta}}^{\star}(\mathbf{w})\hat{\mathbf{G}}_B\mathbf{w}\big|^2}{\sigma_B^2}
\\
\text{s.t.}\quad
& \|\mathbf{w}\|^2 \leq P_S^{\max},P_e^{\star}(\mathbf{w}) \geq 1-\varepsilon.
\end{aligned}
\end{equation}
Although \eqref{snr_opt} is non-convex, both the covert rate and the detection-error-probability constraint are smooth functions of $\mathbf{w}$. Therefore, we solve \eqref{snr_opt} by an SQP method, where in each iteration a quadratic approximation of the Lagrangian and linearized constraints are constructed around the current beamformer iterate, and the resulting quadratic program is solved to update $\mathbf{w}$. After each SQP update, the detection error probability $P_e^{\star}(\mathbf{w})$ is re-evaluated to explicitly verify the covertness constraint $P_e^{\star}(\mathbf{w}) \geq 1-\varepsilon$.

The overall procedure is summarized in Algorithm~\ref{alg:global_ao}. 
A warm-start strategy is adopted, where the warm-start beamformer $\mathbf{w}_{\text{warm}}$ is obtained from Algorithm~\ref{alg1} and then used to initialize the SQP iterations. For each initialization, Algorithm~\ref{alg:global_ao} iteratively updates the transmit beamformer $\mathbf{w}$ by solving the SQP subproblem associated with \eqref{snr_opt}, where in each iteration the IRS phase-shift matrix is implicitly given by the deterministic mapping $\mathbf{\Theta}(\mathbf{w})$ in \eqref{theta_mapping}--\eqref{theta_mapping_element}. The solution pair $(\mathbf{w},\mathbf{\Theta}(\mathbf{w}))$ that achieves the largest covert rate is selected as the final output.

\begin{algorithm}[t]
\caption{SQP-based Beamforming Algorithm for Multi-IRS-aided CPWC under Detection-error Constraint}
\begin{algorithmic}[1]\label{alg:global_ao}
\STATE \textbf{Initialize:}
\STATE Set the tolerance accuracy threshold $\delta$;
\STATE Set the maximum number of iterations (e.g., \texttt{maxIter}).
\STATE \textbf{Begin warm-start procedure:}
\STATE Obtain the warm-start beamformer $\mathbf{w}_{\text{warm}}$ from Algorithm~\ref{alg1};
\STATE \textbf{Initialization for SQP:}
\STATE $\mathbf{w}^{(0)} \leftarrow \mathbf{w}_{\text{warm}}$; 
\STATE Set iteration index $i \leftarrow 0$;
\STATE $R_{\text{old}} \leftarrow -\infty$,
       $R_{\text{new}} \leftarrow R_B\big(\mathbf{w}^{(0)},\mathbf{\Theta}(\mathbf{w}^{(0)})\big)$;
\STATE \textbf{Iterative SQP-based update:}
\WHILE{$|R_{\text{new}} - R_{\text{old}}| > \delta$ and $i < \texttt{maxIter}$}
    \STATE $i \leftarrow i + 1$; $R_{\text{old}} \leftarrow R_{\text{new}}$;
    \STATE Obtain $\mathbf{w}^{(i)}$ by solving problem \eqref{snr_opt} via SQP;
    \STATE Evaluate $P_e^{\star}(\mathbf{w}^{(i)})$ using the closed-form expression in Section~IV-A;
    \STATE Update $R_{\text{new}} \leftarrow R_B\big(\mathbf{w}^{(i)},\mathbf{\Theta}(\mathbf{w}^{(i)})\big)$;
\ENDWHILE
\STATE Select the converged beamformer $\mathbf{w}^{(i)}$ as the final solution.
\RETURN $\mathbf{w}^* \leftarrow \mathbf{w}^{(i)}$, $\hat{\boldsymbol{\theta}}^* \leftarrow \hat{\boldsymbol{\theta}}^{\star}(\mathbf{w}^*)$.
\end{algorithmic}
\end{algorithm}

In terms of computational complexity, the main computational burden of Algorithm~2 comes from the SQP update of $\mathbf{w}$ and the closed-form update of the IRS phases. For the $N$-dimensional beamformer, the SQP step has a complexity on the order of $\mathcal{O}(N^{3.5})$, while updating the IRS phase  only involves element-wise operations on $\hat{\mathbf{G}}_B\mathbf{w}$ with a complexity of $\mathcal{O}(LMN)$. Let $I_{\text{total2}}$ denote the number of outer iterations of Algorithm~2, then the overall computational complexity is approximately
$\mathcal{O}\big(I_{\text{total2}}(N^{3.5}+LMN)\big)$.

\section{Numerical Results}
In this section, numerical results are presented to evaluate the effectiveness of the proposed algorithms.
The main simulation parameters are summarized in Table~\ref{parameter}, which serves as the common configuration for all experiments. 

The satellite-related parameters in Table~I, including the altitude of the LEO satellite, central carrier frequency, bandwidth, the maximum satellite antenna gain, off-axis angle, and 3 dB beam-width angle of the antenna, are selected according to representative 6G NTN studies and satellite-terrestrial convergence technical references~\cite{wigard2023ubiquitous,mediatek2023satellite}.

Additionally, the satellite-IRS link is assumed to experience infrequent light shadowing (ILS) with parameters $(m,b,\Omega )=(10,0.158,1.29)$~\cite{wu2023covert}.
The noise power is computed using $\kappa_B$TB, according to the parameter values listed in Table~\ref{parameter}.
The coordinates of IRS $I_1$ and $I_2$ are set to $(-50~\rm m,-150~\rm m,50~\rm m)$ and $(-100~\rm m,-200~\rm m,30~\rm m)$, respectively.

The runtime results in the convergence analysis are measured on a computer
equipped with an Intel Core Ultra 9 275HX processor.

\renewcommand{\arraystretch}{1.2}
\begin{table}
\caption{Simulation parameters}
\label{parameter}
\centering
\begin{tabularx}{0.45\textwidth}{lX}  
\hline                      
Physical meaning             & Value \\
\hline
Number of satellite antennas &$N=4$\\
Number of IRS elements & $M_x\times M_z=4\times4$   \\
Number of IRS       & $L=2$                  \\
Satellite's location     & (-30 $\rm m$, -400 $\rm m$, 600 $\rm km$)              \\
Bob's location       & (0, 0, 0)             \\
Willie's location           & (0, 50 $\rm m$, 0)   \\
Central carrier frequency      & $f_c=3~\rm GHz$            \\
Number of subcarriers           & $N_s=1024$       \\
Bandwidth              & $B=5~\rm MHz$         \\
Boltzmann's constant         & $\kappa_B=1.38\times10^{-23}$~\text{J/K} \\ 
Noise temperature              & $T=300~\rm K$      \\
Maximum antenna gain of satellite  &$G_S^{\max}=49.5~\rm dBi$\\
Off-axis angle  & $\gamma_{\mathrm{off}}=\ang{0.2}$\\
3 \text{dB} beam-width angle
antenna  &$\gamma _{3dB}=\ang{3.2}$\\
Rician factor        & $K=10~\rm dB$      \\   
Path-loss at the reference distance                        & $\mathcal{L}_{0}=10^{-3}$      \\ 
Path loss exponent                           & $\zeta=2.2$            \\
\hline
\end{tabularx}
\end{table}

To validate the effectiveness of the proposed multi-IRS-aided CPWC schemes, including the AO--SDR Algorithm and the SQP-based beamforming design, we compare their performance against the following benchmark schemes:

\begin{itemize}
\item{Upper bound}: The upper bound on the covert rate is obtained from the objective value of problem~\eqref{Problem1.2.1}, which serves as a theoretical performance limit for Algorithm~\ref{alg1}.
\item{Random phase shift}: The phase shifts of all IRS elements are randomly set within $[0, 2\pi]$, and the transmit beamforming vector is optimized using the proposed Algorithm~\ref{alg:global_ao}.
\item{SOCP}: The SOCP-based algorithm described in Section~III is employed for joint optimization of the transmit beamforming and IRS phase shifts.
\item{1-bit IRS}: A practical scenario with discrete phase shifts is considered, where the IRS phase shifts obtained from Algorithm~\ref{alg:global_ao} are quantized via the NPP method. 
\item \begingroup\rightskip=0pt plus 1em
{Without PC~\cite{10035943}}: An IRS-assisted satellite-terrestrial covert communication scheme is implemented, where SDR is used to optimize IRS phase shifts under a path-loss-only channel model, and the single IRS in the without PC baseline is configured with the same total number of elements as all IRSs in our multi-IRS scheme to ensure a fair comparison.
\par\endgroup
 
\end{itemize}

Unless otherwise specified, the proposed scheme refers to Algorithm~\ref{alg:global_ao}, which uses the relative-entropy solution from Algorithm~\ref{alg1} as a warm start. Algorithm~\ref{alg:global_ao} refines the beamformer via SQP under the detection-error-probability constraint, yielding a higher covert rate than the SOCP benchmark.

Fig.~\ref{precise} demonstrates the 3D performance surface of the SNR versus the user location, where only the $\rm X$- and $\rm Y$-axis coordinates are considered, as both Bob and Willie are ground users.
It can be seen that the SNR reaches its maximum at Bob's position (0, 0, 0), forming a distinct energy peak.
In contrast, the SNR at Willie's location (0, 50 $\rm m$, 0) is as low as 0.3, thereby verifying the effectiveness of the proposed method in achieving covert precision communication.
Moreover, the SNR at (-50 $\rm m$, -100 $\rm m$, 0) reaches a value of 3, which can be attributed to the proximity of this location to IRS $I_2$, resulting in enhanced reflected signal strength.

\begin{figure}
\centering
\includegraphics[width=0.35\textwidth]{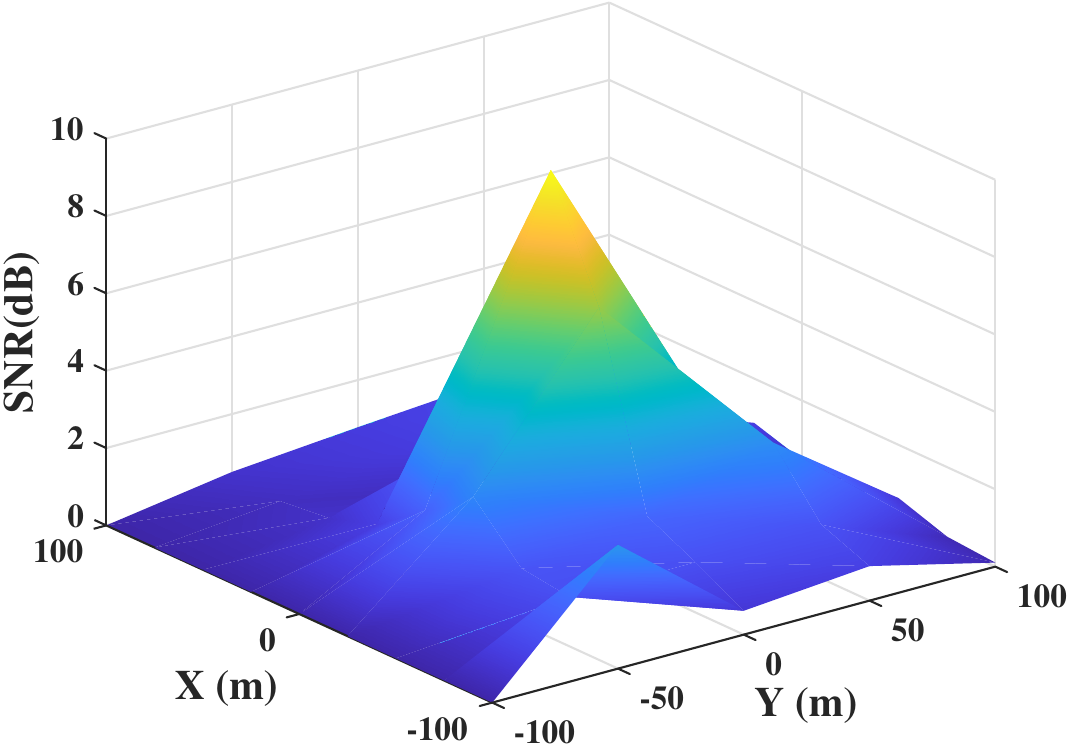}
\caption{3D performance surface of SNR versus the location coordinate under the proposed algorithm.}
\label{precise}
\end{figure}

Fig.~\ref{fig:convergence_behavior} shows the convergence behavior of the proposed SQP-based refinement scheme under different satellite transmit powers.
Algorithm~2 is initialized by the warm start solution generated from Algorithm~1 and further optimizes the beamformer under the detection-error-probability constraint.
For all considered transmit powers, the covert rate increases rapidly during the first few SQP iterations and reaches convergence after about four iterations.
After convergence, the final covert rate exhibits a clear dependence on the satellite transmit power. Specifically, a higher satellite transmit power leads to a higher converged covert rate, while the convergence speed remains almost unchanged across different power settings. Moreover, for Algorithm~2, the effective execution times are $3.35~\mathrm{s}$, $6.02~\mathrm{s}$, $5.89~\mathrm{s}$, and $4.83~\mathrm{s}$ under $P_S^{\max}=49$, $51$, $53$, and $55~\mathrm{dBm}$, respectively.
The average effective execution time is about $5.02~\mathrm{s}$.
Together with the complexity analysis of Algorithm~2, the convergence and runtime results show that the proposed refinement scheme achieves fast convergence in terms of both the required iteration number and the measured execution time.

\begin{figure}[t]
    \centering
    \includegraphics[width=0.7\linewidth]{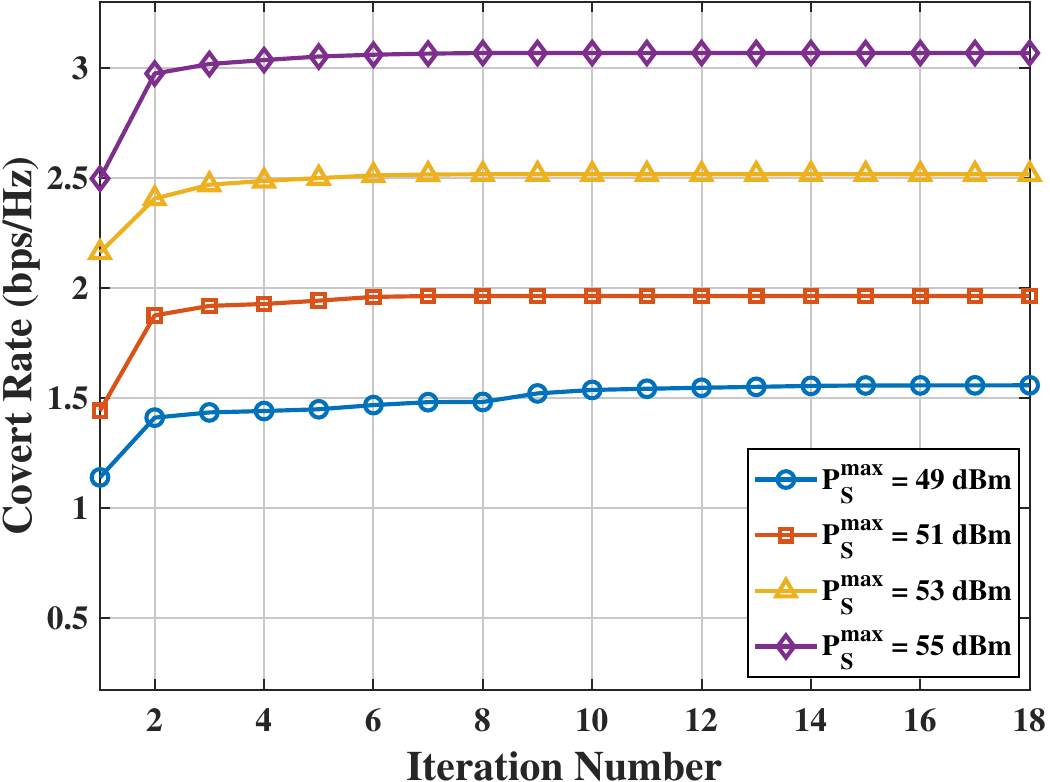}
    \caption{Convergence analysis of the proposed SQP-based refinement scheme under different satellite transmit powers.}
    \label{fig:convergence_behavior}
\end{figure}

\begin{figure*}[!t]
\centering
\begin{minipage}[t]{0.32\textwidth}
  \centering
  \includegraphics[width=\textwidth]{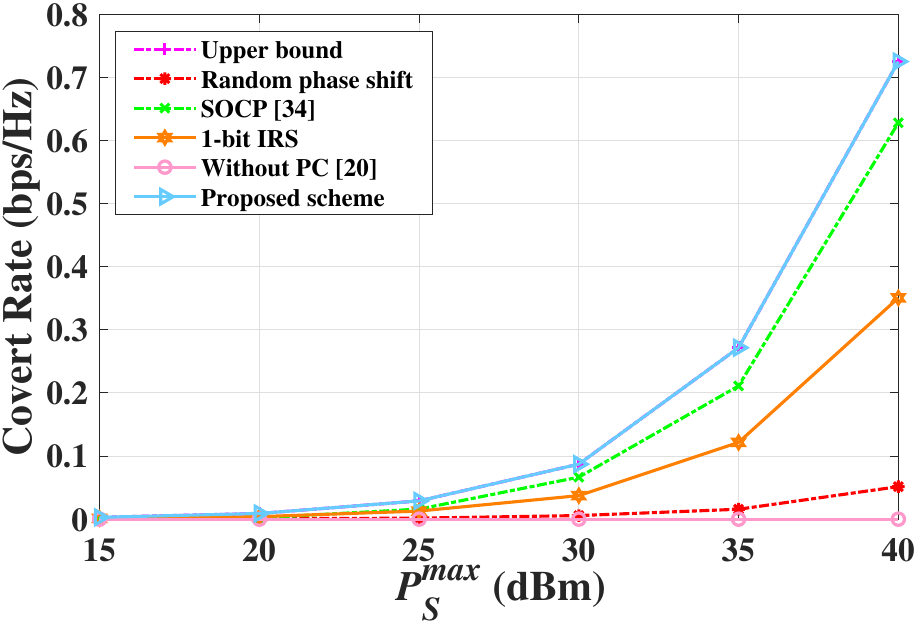}
  \caption{Covert rate versus the maximum transmit power $P_S^{\max}$ at the satellite with covertness performance $\varepsilon = 0.05$.}
  \label{Pa_rate}
\end{minipage}
\hfill
\begin{minipage}[t]{0.32\textwidth}
  \centering
  \includegraphics[width=\textwidth]{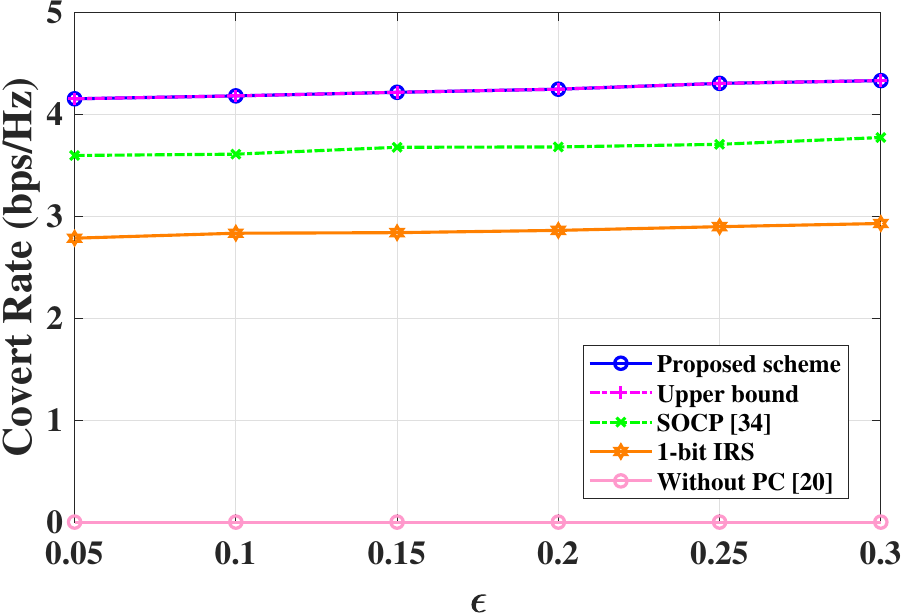}
  \caption{Covert rate versus covertness performance $\varepsilon$ with different $P_S^{\max}$.}
  \label{epsilon_rate}
\end{minipage}
\hfill
\begin{minipage}[t]{0.32\textwidth}
  \centering
  \includegraphics[width=\textwidth]{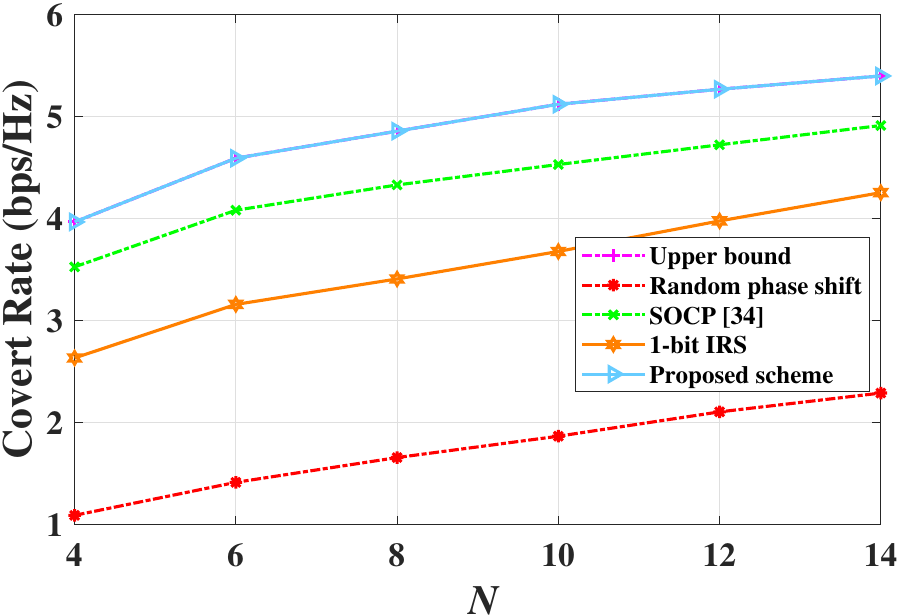}
  \caption{Covert rate versus the number of satellite antennas $N$ with $\varepsilon = 0.05$ and $P_S^{\max} = 55~\rm dBm$.}
  \label{N_rate}
\end{minipage}
\end{figure*}

Fig.~\ref{Pa_rate} illustrates the variation curves of the covert rate versus the maximal transmit power $P_S^{max}$ with $\varepsilon=0.05$, where different optimization schemes are compared. 
It can be observed that the covert rate monotonically increases with $P_S^{\max}$.
Compared with the random phase shift scheme, the covert rate achieved by other optimization schemes grows rapidly as $P_S^{max}$ increases from 30~$\text{dBm}$ to 40~$\text{dBm}$, since the absence of reflective beamforming results in a slower rate improvement.
Furthermore, the SDR upper bound of \eqref{Problem1.2.1} serves as a reference. Algorithm~\ref{alg1} achieves a performance close to this bound with Gaussian randomization, while Algorithm~\ref{alg:global_ao}, further improves the covert rate through SQP refinement under the detection-error-probability constraint.
Moreover, the covert rate decreases when a 1-bit IRS is applied instead of an ideal IRS, due to the lack of phase shift accuracy.
Finally, even when the number of IRS elements in the SDR scheme without PC is set equal to the total number of IRS elements in the proposed scheme, the covert rate achieved by the SDR scheme remains close to zero, thereby demonstrating the effectiveness of the proposed approach.

As shown in Fig.~\ref{epsilon_rate}, we compare the covert rate against the covertness performance parameter $\varepsilon$ across various optimization schemes.
The covert rate is observed to increase as $\varepsilon$ increases.
This trend occurs because a higher $\varepsilon$ reduces the covertness constraint, allowing more energy to be directed to the legitimate receiver while still ensuring covert communication that remains undetectable by Willie.
Moreover, while our proposed method achieves a higher covert rate than the SOCP method as $P_S^{\max}$ decreases, the gap between them narrows significantly, indicating that our algorithm achieves superior performance when $P_S^{\max}$ is large.

In Fig.~\ref{N_rate}, the impact of the number of satellite antennas $N$ on the covert rate is investigated.
It can be observed that the performance of all schemes exhibits a consistent increasing trend as $N$ grows.
This improvement stems from the enhanced ability of multiple antennas to concentrate signal power in desired directions, thereby improving Bob’s receiving SNR.
On the other hand, by comparing Fig.~\ref{N_rate} with Fig.~\ref{Pa_rate}, it is evident that increasing the number of transmitter antennas is more effective than simply increasing the transmitter power in improving the covert rate, which provides new insights for improving the covert rate under power-constrained conditions.
It is also noted that the rate of performance improvement diminishes once the number of antennas exceeds a certain threshold.

Fig.~\ref{M_rate} shows the impact of the number of elements $M$ in each IRS on the covert rate.
It can be observed that the covert rate increases with $M$, since a large number of elements enables more coherent superposition of reflection paths, thereby enhancing the beamforming gain towards Bob.
Moreover, the growth rates of the different schemes exhibit noticeable differences.
In particular, the 1-bit IRS scheme demonstrates the steepest growth, and when $M$ reaches 64, the covert rate approaches that obtained by the SOCP scheme.
Therefore, in practice, we can compensate for the lack of discrete IRS phase shift accuracy by increasing the number of IRS elements.
Moreover, the performance of the proposed schemes consistently outperforms the benchmarks, achieving an improvement of $3.32\%$ over the SOCP scheme and $3.97\%$ over the 1-bit IRS scheme.

\begin{figure*}[!t]
\centering
\begin{minipage}[t]{0.32\textwidth}
  \centering
  \includegraphics[width=\textwidth]{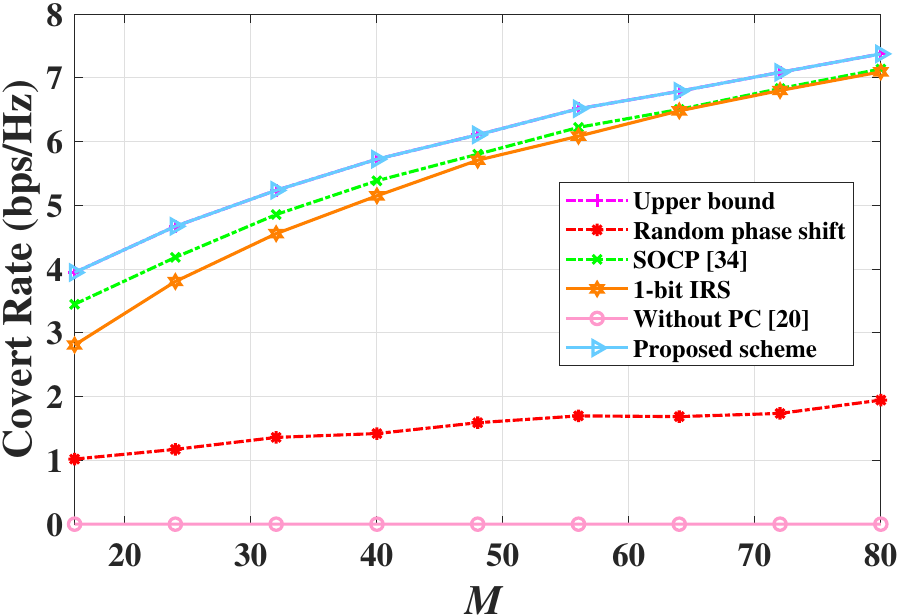}
  \caption{Covert rate versus the number of elements $M$ in each IRS with $\varepsilon = 0.05$ and $P_S^{\max} = 55~\rm dBm$.}
  \label{M_rate}
\end{minipage}
\hfill
\begin{minipage}[t]{0.32\textwidth}
  \centering
  \includegraphics[width=\textwidth]{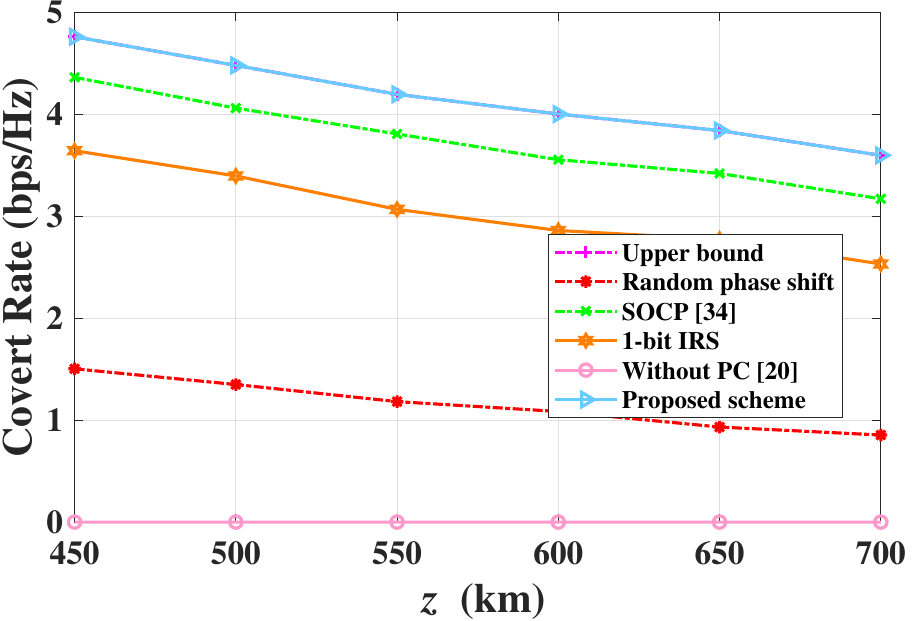}
  \caption{Covert rate versus satellite altitude $z_s$ with $\varepsilon = 0.05$ and $P_S^{\max} = 55~\rm dBm$.}
  \label{h_rate}
\end{minipage}
\hfill
\begin{minipage}[t]{0.32\textwidth}
  \centering
  \includegraphics[width=\textwidth]{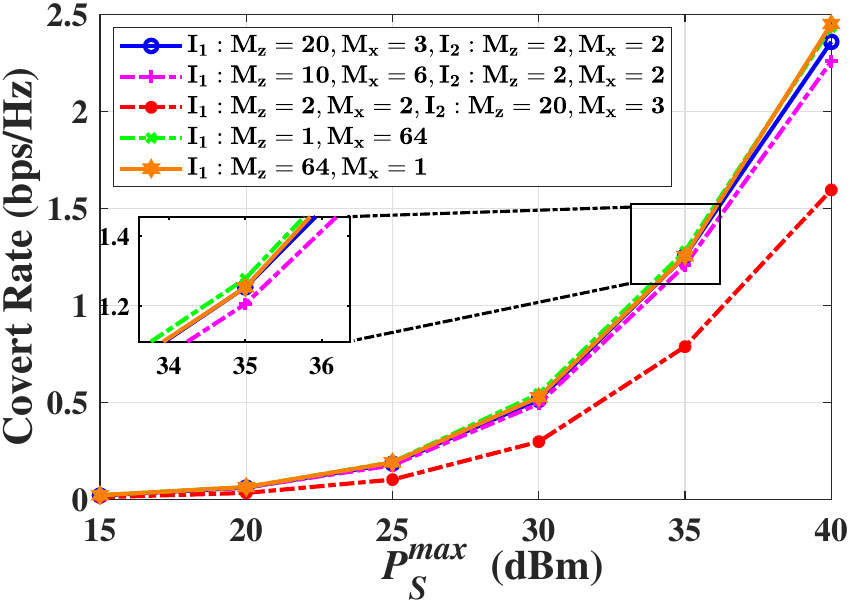}
  \caption{Covert rate versus the maximum transmit power $P_S^{\max}$ with $\varepsilon = 0.05$ and various IRS geometric sizes.}
  \label{IRS}
\end{minipage}
\end{figure*}

Fig.~\ref{h_rate} illustrates the covert rate versus $z_s$ when $\varepsilon=0.05$ and $P_S^{max}=55~\rm dBm$.
As expected, the covert rate monotonically decreases with $z_s$.
This is because the channel gain from the satellite to the ground receiver is dominated by the path loss, where $z_s$ determines the propagation distance. 
An increase in $z_s$ leads to higher path loss, thereby reducing the received signal power.
Moreover, we observe that the rate of decline in the covert rate initially slows down and subsequently accelerates as $z_s$ increases.
Quantitatively, in the proposed Algorithm~\ref{alg:global_ao}, the covert rate decreases by an average of $0.23262~\rm bps/Hz$ for every $50~\rm km$ increase in altitude.
In comparison, the SOCP scheme exhibits an average covert rate decrease of $0.23826~\rm bps/Hz$ for every $50~\rm km$, while the 1-bit IRS scheme experiences a smaller decrease of $0.22206~\rm bps/Hz$ per $50~\rm km$ increase in altitude.

In Fig.~\ref{IRS}, we investigate the impact of the maximum transmit power $P_S^{max}$ on the covert rate under different IRS geometries with $\varepsilon=0.05$.
First, we observe that for a single ULA-structured IRS with the same number of elements, the covert rate remains nearly constant.
Secondly, when the number of elements in $I_{1}$ is set to 60 and that in $I_{2}$ to 4, the covert rate increases with $P_S^{\max}$, since the element arrangement of $I_{1}$ gradually approximates a ULA, thereby providing stronger beamforming gain.
By contrast, when $I_{2}$ is configured with 60 elements while $I_{1}$ is limited to 4, the covert rate is significantly reduced compared to the previous case.
This is due to the fact that $I_{1}$ is positioned closer to the satellite and the legitimate receiver than $I_{2}$, thereby yielding a higher power gain in the former configuration.

\begin{figure}[t]
    \centering
    \includegraphics[width=0.8\linewidth]{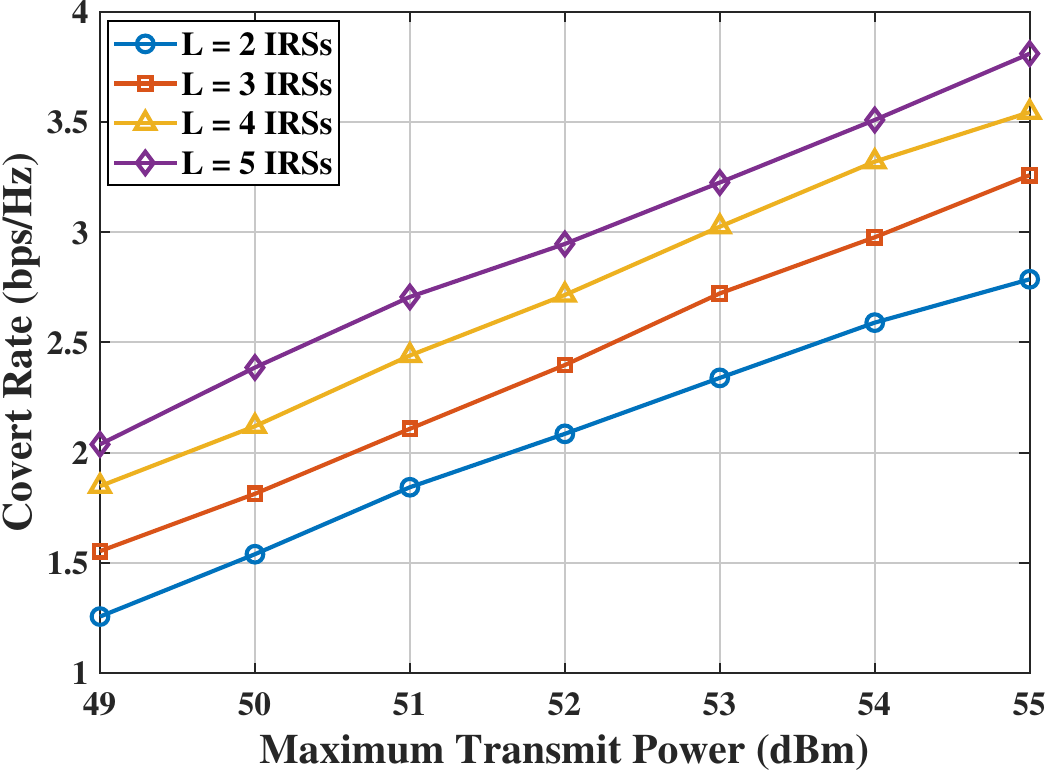}
    \caption{Covert rate versus maximum satellite transmit power under different numbers of distributed IRSs.}
    \label{fig:multi_irs_cr}
\end{figure}

Fig.~\ref{fig:multi_irs_cr} shows the covert rate versus the maximum satellite transmit power under different numbers of distributed IRSs.
For all considered transmit-power values, a larger number of IRSs yields a higher covert rate.
As $P_S^{\max}$ increases from $49~\mathrm{dBm}$ to $55~\mathrm{dBm}$, the covert rate increases monotonically for all considered IRS numbers.
Specifically, the covert-rate range shifts upward from about $1.25$-$2.04~\mathrm{bps/Hz}$ at $49~\mathrm{dBm}$ to about $2.78$-$3.82~\mathrm{bps/Hz}$ at $55~\mathrm{dBm}$ when $L$ increases from $2$ to $5$.
The performance improvement comes from the additional controllable reflecting paths and larger passive reflecting aperture provided by more distributed IRSs.
With a larger $L$, the joint beamforming and IRS phase-shift optimization can enhance Bob-side coherent signal combining while regulating the leakage toward Willie under the covertness constraint.
As a result, the proposed framework is not limited to the two-IRS case and can benefit from larger-scale distributed IRS deployment.

\begin{figure}[t]
    \centering
    \includegraphics[width=0.8\linewidth]{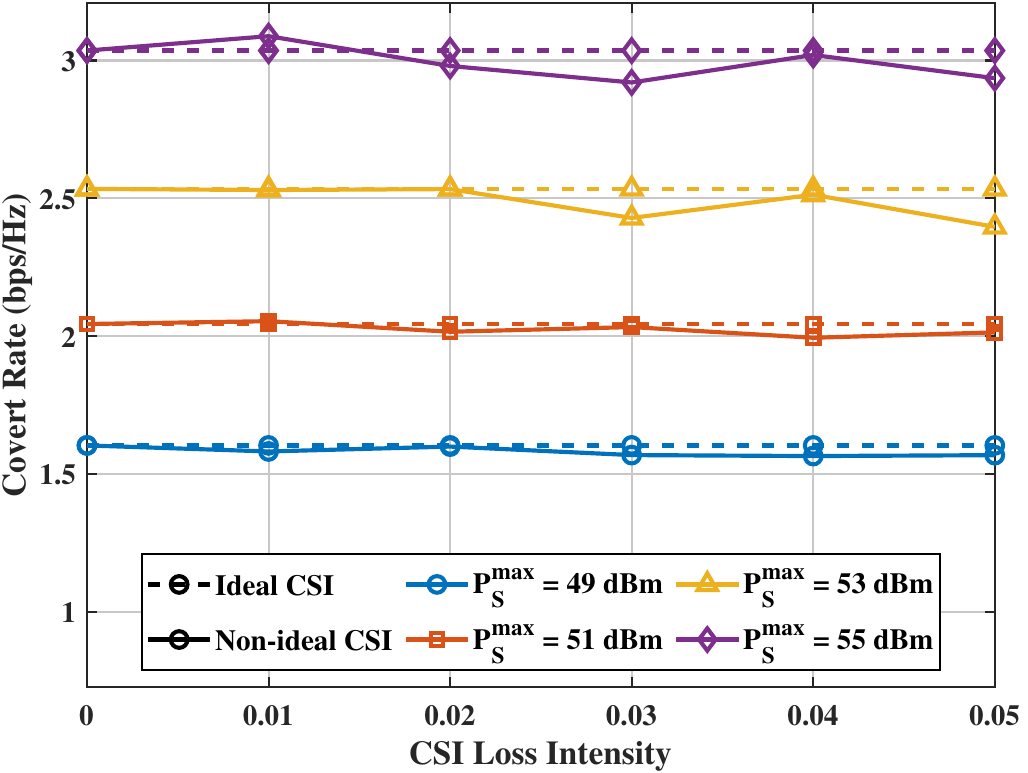}
    \caption{Covert rate under different CSI loss intensities.}
    \label{fig:csi_loss}
\end{figure}

For the non-ideal CSI case, channel estimation errors, feedback delay, and IRS reconfiguration latency may cause mismatch between the optimized beamforming and IRS phase-control coefficients and the actual cascaded channel, leading to Bob-side focusing loss and covert-rate reduction. As a result, a quantitative sensitivity experiment is added to evaluate the impact of non-ideal CSI on the proposed CPWC system.

Fig.~\ref{fig:csi_loss} evaluates the sensitivity of the proposed scheme to non-ideal CSI.
The dashed curves denote the ideal CSI benchmark, while the solid curves denote the non-ideal CSI case with different CSI loss intensities.
As the CSI loss intensity increases from $0$ to $0.05$, the covert rate remains close to the ideal CSI benchmark under all considered transmit powers.
For $P_S^{\max}=49~\mathrm{dBm}$ and $51~\mathrm{dBm}$, the non-ideal CSI curves almost overlap with the ideal CSI references, indicating negligible performance degradation under mild CSI mismatch.
For $P_S^{\max}=53~\mathrm{dBm}$ and $55~\mathrm{dBm}$, the non-ideal CSI curves exhibit slightly larger fluctuations, but the rate loss remains moderate and no obvious performance collapse occurs.
Moreover, the performance ordering under different transmit powers is preserved, since a larger $P_S^{\max}$ consistently yields a higher covert rate.
The above observations indicate that the proposed scheme maintains stable performance under mild CSI loss and provides a quantitative sensitivity evaluation beyond the ideal CSI benchmark.

\section{Conclusions}
In this paper, we have proposed a multi-IRS-aided satellite CPWC system, where an LEO satellite serves a legitimate ground user while mitigating detection attempts by a warden.
By integrating the OFDM-based RSCS and beamforming, the transmit energy has been effectively concentrated toward the legitimate receiver, thereby enhancing communication security.
Based on the detection error probability and relative entropy, the closed-form expressions for the covertness of the proposed CPWC system have been derived under Nakagami-m fading for the satellite–IRS link and Rician fading for the IRS–user link.
To maximize the covert rate, we have formulated an optimization problem under the CPWC covertness constraints and the satellite energy budget. First, an AO-based SDR iterative algorithm has been developed under the relative-entropy constraint to obtain a high-quality feasible warm start. Using this warm start, we have imposed the detection-error-probability constraint and applied an SQP-based refinement to further improve the covert rate.
Additionally, a discrete phase shift scheme for the IRS has been designed to extend the proposed system to practical implementations.
Numerical results validate the effectiveness of the proposed schemes in achieving CPWC and demonstrate the efficiency and superiority of the developed algorithm.

Future work will extend the terrestrial-IRS-aided CPWC framework toward a robust and delay-aware design by accounting for non-ideal cascaded CSI, IRS reconfiguration latency, and LEO mobility in the joint beamforming and IRS phase-control design. 
In addition, the CPWC framework can be expanded to satellite-side IRS deployment, where space implementation constraints need to be integrated with covert and precise transmission design to support practical on-orbit deployment.







\appendices
\section{Proof of Theorem~\ref{thm1}}\label{proof of thm1}
To prove Theorem~\ref{thm1}, we note that Problem~\eqref{Problem1.1.1} is convex and satisfies Slater's condition, so strong duality holds and the optimal solution can be obtained from its Lagrangian dual problem. According to~\eqref{Problem1.1.1}, the corresponding Lagrangian function can be expressed as
\begin{equation}
\begin{aligned}
 &\mathcal{L}\left(\mathbf{W}, \mathbf{X}, \alpha_1, \alpha_2\right) = -\text{tr}\left(\mathbf{C_b} \mathbf{W}\right) - \text{tr}\left(\mathbf{X} \mathbf{W}\right) \\
& + \alpha_1\left(\text{tr}\left(\mathbf{C_w} \mathbf{W}\right) - f^{-1}(2\varepsilon^2) + 1 \right) + \alpha_2\left( \text{tr}\left(\mathbf{W}\right) - P_s^{\max} \right),\\
\end{aligned}    
\end{equation}
where $\alpha_1$ and $\alpha_2$ represent Lagrange multipliers w.r.t. constraints~\eqref{Problem1.1.0a} and~\eqref{Problem1.1.0b}, respectively, and $\mathbf{X}$is Lagrange multiplier matrix of~\eqref{Problem1.1.0c}.
Accordingly, the Lagrangian dual problem can be expressed as
\begin{equation*}
\underset{\alpha_1\geq0,\alpha_2\geq0,\mathbf{X}\succeq0}{\text{max}}
\underset{\mathbf{W}}{\text{min}}\quad
\mathcal{L}\left(\mathbf{W}, \mathbf{X}, \alpha_1, \alpha_2\right),
\end{equation*}
which can be solved to obtain the optimal solution
\begin{equation*}
\left(\mathbf{W}^*, \mathbf{X}^*, \alpha_1^*, \alpha_2^*\right)
\end{equation*}
to the primal-dual problem, given by
\begin{align}
\left(\mathbf{W}^*, \mathbf{X}^*, \alpha_1^*, \alpha_2^*\right)=\text{arg} \left \{\underset{\alpha_1,\alpha_2,\mathbf{X}}{\text{max}}\underset{\mathbf{W}}{\text{min}}\quad\mathcal{L}\left(\mathbf{W}, \mathbf{X}, \alpha_1, \alpha_2\right) \right \} . 
\end{align}

According to~\cite{chi2017convex}, $\left(\mathbf{W}^*, \mathbf{X}^*, \alpha_1^*, \alpha_2^*\right)$ satisfies the Karush-Kuhn-Tucker (KKT) conditions, which are explicitly given as
\begin{equation}
\begin{aligned}
\mathrm{K1:} \quad & \nabla_{\mathbf{W}^*}\mathcal{L}\left(\mathbf{W}^*, \mathbf{X}^*, \alpha_1^*, \alpha_2^*\right)=0, \\
\mathrm{K2:} \quad & \eqref{Problem1.1.0a},\eqref{Problem1.1.0b},\eqref{Problem1.1.0c}, \\
\mathrm{K3:} \quad & \alpha_1^*\geq0,\alpha_2^*\geq0,\mathbf{X}^*\succeq0,\\
\mathrm{K4:} \quad & \alpha_1^*\left(\text{tr}\left(\mathbf{C_w} \mathbf{W}^*\right) - f^{-1}(2\varepsilon^2) + 1 \right)=0,\\
\mathrm{K5:} \quad & \alpha_2^*\left( \text{tr}\left(\mathbf{W}^*\right) - P_s^{\max} \right)=0,\\
\mathrm{K6:} \quad & \text{tr}\left(\mathbf{X}^* \mathbf{W}^*\right) = 0.
\end{aligned}
\end{equation}

In the following, we prove that \text{rank}$(\mathbf{W}^*)\leq1$ via analyzing the structure of $\mathbf{X}^*$.
First, $\mathrm{K1}$ can be re-expressed as
\begin{align}
\label{X^*}
 \mathbf{X}^*=\alpha_2^* \mathbf{I}_{N} - \left( \mathbf{C_b} - \alpha_1^*\mathbf{C_w} \right).  
\end{align}

Define $\mathbf{\Xi}^*=\mathbf{C_b} - \alpha_1^*\mathbf{C_w}$, and let the maximum eigenvalue of $\mathbf{\Xi}^*$ be denoted by $\lambda_{\text{max}}$.
From~\eqref{X^*} together with K3, the feasibility condition $\mathbf{X}^* \succeq 0$ requires $\alpha_2^* \geq \lambda_{\text{max}}$. Under this condition, all eigenvalues of $\mathbf{X}^*$ are nonnegative and at most one eigenvalue can be zero for random channel realizations~\cite{9183907}, which implies $\text{rank}(\mathbf{X}^*) \geq N-1$.
Moreover, K6 gives $\text{tr}\!\left(\mathbf{X}^* \mathbf{W}^*\right) = 0$. Since $\mathbf{X}^* \succeq 0$ and $\mathbf{W}^* \succeq 0$, the equality $\text{tr}\!\left(\mathbf{X}^* \mathbf{W}^*\right) = 0$ yields $\mathbf{X}^* \mathbf{W}^* = 0$, and the rank inequality leads to $\text{rank}(\mathbf{X}^*)+\text{rank}(\mathbf{W}^*)\leq N$, so that $\text{rank}(\mathbf{W}^*)\leq 1$.
Since $\mathbf{W}^*$ is the transmit covariance matrix, the trivial solution $\mathbf{W}^*=\mathbf{0}$ is infeasible in practice, and therefore $\text{rank}(\mathbf{W}^*)=1$, which completes the proof of Theorem~\ref{thm1}.

\section{Proof of Theorem~\ref{thm2}}\label{proof of thm2}
Considering that only the statistical CSI of $\mathbf{h}_{S, I_l}$ is available at Willie, we define the received power as $Y=\left |\sum_{l=1}^{L}\mathbf{h}_{I_{l},W}^{H}\mathbf{\Phi}_{l}\mathbf{H}^{H}_{S,I_{l}}\mathbf{w}\right |^2$.
To characterize the distribution of $Y$, we first derive the distribution of the intermediate variable $\mathcal{Q}=\sum_{l=1}^{L}\mathbf{h}_{I_{l},W}^{H}\mathbf{\Phi}_{l}\mathbf{H}^{H}_{S,I_{l}}\mathbf{w}$.

Let $[\cdot ]_i$ be the $i$-th element of a vector.
Substitute $\mathcal{Q}$ into~\eqref{channel1}, we denote $\mathcal{K}_i=\lambda \frac{\sqrt{G_S}}{4\pi d_{S,I_l}}[\mathbf{a}_{S,I_l}]_i[\mathrm{diag}(\mathbf{w})\mathbf{H}_{I_l,W}^H\theta_l]_i$.
Let $z_i=[\mathbf{z}]_i$, and $\mathcal{Q}$ can be expressed as
\begin{align}\label{S}
\mathcal{Q}=\displaystyle\sum_{l=1}^{L}\displaystyle\sum_{i=1}^{N}|z_i||\mathcal{K}_i|e^{j\varphi_{z_i} }e^{j\varphi_{\mathcal{K}_i} } , 
\end{align}
where $\varphi_{z_i}$ and $\varphi_{\mathcal{K}_i}$ represent the phases of $z_i$ and $\mathcal{K}_i$, respectively.

In multi-IRS-assisted systems, ideal cascaded phase alignment is commonly assumed~\cite{9205879}.
Specifically, we impose the cascaded phase-alignment constraint $\varphi_{z_i} + \varphi_{\mathcal{K}_i} = 0$ for each element $i$ by configuring the IRS phase-shift vector $\boldsymbol{\theta}_l$ as a function of the transmit beamformer $\mathbf{w}$, so that all reflected components are phase-aligned at the receiver.
Under the above phase-alignment assumption, the received signal power simplifies as
\begin{align}\label{Y}
Y=|\mathcal{Q}|^2=(\displaystyle\sum_{i=1}^{NL}|z_i||\mathcal{K}_i|)^2.  
\end{align}
Given $|z_i|\sim \text{Nakagami}(m,\Omega)$, the $k$-th moment of $|z_i|$, denoted as $\mu_{|z_i|}(k)\triangleq \mathbb{E}[|z_i|^k] $, is given by
\begin{align}\label{mu}
\mu_{|z_i|}(k)=\frac{\Gamma(m+k/2)}{\Gamma(m)}(\frac{m}{\Omega})^{-k/2}. 
\end{align}
From~\eqref{Y}, $Y$ consists of terms involving squared Nakagami-distributed variables and products of Nakagami-distributed variables. Since the square of a Nakagami variable follows a Gamma distribution and the product of two independent Nakagami variables can be approximated by a Gamma distribution~\cite{9558795}, the distribution of $Y$ can be approximated by a Gamma distribution.

In the following, we derive the parameters of $Y$, following the gamma distribution.
According to the Central Limit Theorem, when $NL$ is sufficiently large, the random variable $\mathcal{Q}$ can be approximated by a Gaussian distribution, i.e. $\mathcal{Q}\sim\mathcal{N}(\mu_S,\sigma_S^2)$.

To compute the mean and variance, we utilize the first and second moments of $|z_i$. Specifically, from~\eqref{mu}, we obtain the  $\mu_{|z_i|}(1)=\mathbb{E}[|z_i|]=\frac{\Gamma(m+0.5)}{\Gamma(m)}\sqrt{\frac{\Omega}{m}}$ and $\mu_{|z_i|}(2)=\Omega$, which yield the variance $\text{VAR}[|z_i|]=\Omega-\mu^2$, where $\mu=\mu_{|z_i|}(1)$.

Therefore, the expected value and variance of $|\mathcal{Q}|$ are then given by
\begin{align}
\mu_{|\mathcal{Q}|}= \mathbb{E}[|\mathcal{Q}|]=\displaystyle\sum_{i=1}^{NL}|\mathcal{K}_i|\mathbb{E}[|z_i|]=\displaystyle\sum_{i=1}^{NL}|\mathcal{K}_i|\mu,
\end{align}
and
\begin{align}
\sigma_{|\mathcal{Q}|}^2=\text{VAR}[|\mathcal{Q}|]=\displaystyle\sum_{i=1}^{NL}|\mathcal{K}_i|^2(\Omega-\mu^2).
\end{align}
Accordingly, the first two moments of $Y=|\mathcal{Q}|^2$ can be derived as
\begin{align}
\mathbb{E}[Y]= \mathbb{E}[|\mathcal{Q}|^2]=\mathbb{E}[|\mathcal{Q}|]^2+\text{VAR}[|\mathcal{Q}|]=\mu_{|\mathcal{Q}|}^2+\sigma_{|\mathcal{Q}|}^2,
\end{align}
and
\begin{align}\label{Vary}
\text{VAR}[Y]= \mathbb{E}[|\mathcal{Q}|^4] -\mathbb{E}[|\mathcal{Q}|^2]^2.
\end{align}
To compute $\mathbb{E}[|\mathcal{Q}|^4]$, we apply the fourth-moment formula for a normally distributed variable
\begin{align}\label{mathcal{Q}}
\mathbb{E}[|\mathcal{Q}|^4]=3\sigma_{|\mathcal{Q}|}^4+6\mu_{|\mathcal{Q}|}^2\sigma_{|\mathcal{Q}|}^2+\mu_{|\mathcal{Q}|}^4.    
\end{align}
Substituting~\eqref{mathcal{Q}} into~\eqref{Vary} yields
\begin{equation*}
\text{VAR}[Y]=2\sigma_{|\mathcal{Q}|}^4+4\mu_{|\mathcal{Q}|}^2\sigma_{|\mathcal{Q}|}^2.
\end{equation*}

Assuming that $Y\sim \text{Gamma}(\alpha,\beta)$, we have $\mathbb{E}[Y]=\alpha\beta$ and $\text{VAR}[Y]=\alpha\beta^2$.
By matching moments, the parameters $\alpha$ and $\beta$ can be calculated as
\begin{align}\label{alpha}
\alpha=\frac{(\mathbb{E}[Y])^2}{\text{VAR}[Y]}=\frac{(\mu_S^2+\sigma_S^2)^2}{2\sigma_S^4+4\mu_S^2\sigma_S^2},   
\end{align}
and 
\begin{align}\label{beta}
\beta=\frac{\text{VAR}[Y]}{\mathbb{E}[Y]}=\frac{2\sigma_S^4+4\mu_S^2\sigma_S^2}{\mu_S^2+\sigma_S^2},
\end{align}
respectively, where $\mu_S=\mu_{|\mathcal{Q}|}$ and $\sigma_S=\sigma_{|\mathcal{Q}|}$.
Finally, substituting~\eqref{alpha} and~\eqref{beta} into~\eqref{f_Y} completes the proof of Theorem 2.

\small
\bibliographystyle{IEEEtran}
\bibliography{main}



\newpage

 




\vfill

\end{document}